\documentclass[12pt,journal,onecolumn,draftcls]{IEEEtran}

\newcommand{\textred}[1]{\textcolor{red}{#1}}

\ifx\noeditingmarks\undefined
  \newcommand{\pgwrapper}[2]{\textred{#1: #2}}
\else
  \newcommand{\pgwrapper}[2]{}
\fi
\usepackage{soul}
\usepackage{framed}
\usepackage{bbm}
\usepackage{epsfig}
\usepackage{times}
\usepackage{float}
\usepackage{afterpage}
\usepackage{amsmath}
\usepackage{amstext}
\usepackage{amssymb,bm}
\usepackage{latexsym}
\usepackage{color}
\usepackage{graphicx}
\usepackage{amsthm}
\usepackage[center]{caption}
\usepackage{subfigure}
\usepackage{pstricks}
\usepackage{caption}
\usepackage{booktabs}
\usepackage{multicol}
\usepackage{lipsum}
\usepackage{verbatim}
\usepackage{mathabx}
\usepackage{hyperref}

\newtheorem{remark}{Remark}

\newtheorem{lemma}{Lemma}

\newtheorem{theorem}{Theorem}
\newtheorem{example}{Example}

\newcommand{\PP}{\mathbb{P}}

\newcommand{\E}{\mathbbm{E}}
\newcommand{\var}{\mathrm{Var}}

\newcommand{\py}{\widebar{P_{Y^n}}}
\newcommand{\pyl}{\widebar{P_{Y^n}^{(\lambda)}}}
\newcommand{\qim}{Q^n_{i,m}}
\newcommand{\qiml}{{Q_{i,m}^n}_{\lambda_i}}
\newcommand{\hsig}{\hat{\sigma}}
\newcommand{\Hsig}{\hat{\Sigma}}

\newcommand{\widesim}[2][1.5]{
  \mathrel{\overset{#2}{\scalebox{#1}[1]{$\sim$}}}
}
\newcommand{\simiid}{\widesim{\text{i.i.d}}}

\begin{document}

\title{The Strongly Asynchronous  Massive Access Channel}
\author{
\IEEEauthorblockN{
Sara Shahi, Daniela Tuninetti and Natasha Devroye\\}
\IEEEauthorblockA{%
University of Illinois at Chicago, Chicago IL 60607, USA. \\
Email: {\tt sshahi7, danielat, devroye @uic.edu}}%
}

\maketitle

\begin{abstract}
This paper considers a Strongly Asynchronous and Slotted Massive Access Channel (SAS-MAC) where $K_n:=e^{n\nu}$ different users transmit a randomly selected message among $M_n:=e^{nR}$ ones within a strong asynchronous window of length  $A_n:=e^{n\alpha}$ blocks, where each block lasts $n$ channel uses. 
A global probability of error is enforced, ensuring that all the users' identities and messages are correctly identified and decoded. Achievability bounds are derived for the case that different users have similar channels, the case that users' channels can be chosen from a set which has polynomially many elements in the blocklength $n$,  and the case with no restriction on the users' channels.  A general converse bound on the capacity region and a converse bound on the maximum growth rate of the number of users are derived. 
\end{abstract}

\section{Introduction}
\label{sec:intro}
In the Internet of Things (IoT) paradigm it is envisioned that many types of 
devices will be wirelessly connected. 
A foundational study to understand the fundamental tradeoffs and thus enable the successful deployment of ubiquitous, interconnected wireless network is needed. This new paradigm imposes new traffic patterns on the wireless network. Moreover,   devices within such a network often have strict energy consumption constraints, as they are often battery powered  sensors  transmitting bursts of data very infrequently to an access point. 
Finally, as the name suggests, these networks must support a huge number of  inter-connected devices.  

Due to these new network characteristics, we propose  a novel communication and multiple-access model: the Strongly Asynchronous Slotted Massive Access (SAS-MAC). In a SAS-MAC, the number of users $K_n:=e^{n\nu}$ increases exponentially with blocklength $n$ with {\it occupancy} exponent $\nu\geq 0$. Moreover, the users are strongly asynchronous, i.e., they transmit in one randomly chosen time slot within a window of length $A_n:=e^{n\alpha}$ blocks, each block of length $n$, where $\alpha\geq 0$ is the {\it asynchronous}  exponent. In addition, when active, each user can choose from a set of $M_n:=e^{nR}$ messages to transmit.
All transmissions are sent to an access point and the receiver is required to jointly decode and identify {\it all} users. The goal is to characterize the set of all achievable $(R, \alpha, \nu)$ triplets.

\subsection{Past work}
Strongly asynchronous communications were first introduced in~\cite{optimal_sequential} for synchronization of a single user, and later extended in~\cite{aslan-suboptimality} for synchronization with positive transmission rate. 
 
In~\cite{aslan-remark-it} the authors of~\cite{aslan-suboptimality} made a brief remark about a ``multiple access collision channel'' extension of their original single-user model. 
In this model, any collision of users (i.e., users who happen to transmit in the same block) is assumed to result in  output symbols that appear as if generated by noise. The error metric is taken to be the {\it per user} probability of error, which is required to be vanishing  for all but a vanishing fraction of users. 
In this scenario, it is fairly easy to quantify the capacity region for the case that the number of users are less than the square root of the asynchronous window length  (i.e., in our notation $\nu < \alpha/2$). However, finding the capacity of the ``multiple access collision channel'' for {\it global / joint} probability of error, as opposed to per user probability of error, is much more complicated and requires novel achievability schemes and novel analysis tools. 
This is the main subject and contribution of this paper.

Recently, motivated by the emerging machine-to-machine type communications and sensor networks, a large body of work has studied ``many-user'' versions of classical multiuser channels as pioneered 
 in~\cite{chen2017capacity}. 
In~\cite{chen2017capacity} the number of users is allowed to grow linearly with blocklength $n$. A full characterization of the capacity of the synchronous Gaussian (random) many access channel was given~\cite{chen2017capacity}. 
In~\cite{polyanskiy_ManyAccess}, the author studied the synchronous massive random access channel where the total number of users  increases linearly with the blocklength $n$.  However, the users are restricted to use the same codebook and only a per user probability of error is enforced.  In the model proposed here, the users are strongly asynchronous, the number of users grow exponentially with blocklength, and we enforce a global probability of error.

Training based synchronization schemes (the use of pilot signals) was proven to be suboptimal for bursty communications in~\cite{aslan-suboptimality}.  Rather, one can utilize the users' statistics at the receiver for synchronization or user identification purposes. The identification problem (defined in~\cite{kiefer1968sequential}) is a classic problem considered in hypothesis testing. In this problem, a finite number of distinct sources each generates a sequence of i.i.d. samples. The problem is to find the  underlying distribution of each sample sequence, given the constraint that each sequence is generated by a distinct distribution. 

Studies on identification problems all assume a fixed number of sequences. In~\cite{ahlswede2006logarithmically}, authors study  the Logarithmically Asymptotically Optimal (LAO) Testing of identification problem for a  finite number of distributions. In particular,  the identification of only two different objects has been studied in detail,  and one can find the reliability matrix, which consists of the error exponents of all error types. Their optimality criterion is to find the largest error exponent  for a set of error types for given values of the other error type exponents. The same problem with a different optimality criterion was also studied  in~\cite{unnikrishnan2015asymptotically}, where multiple, finite sequences were matched to the source distributions. More specifically, the authors in~\cite{unnikrishnan2015asymptotically} proposed a test for a generalized Neyman-Pearson-like optimality criterion to minimize the rejection probability given that all other error probabilities decay exponentially with a pre-specified slope. 

In this paper, we too allow the number of users to increase in the blocklength. 
We assume that the users are strongly asynchronous and may transmit randomly anytime within a time window that is exponentially large in the blocklength. We require the receiver to recover both the transmitted messages and the users' identities under a global/joint probability of error criteria. By allowing the number of sequences to grow exponentially with the number of samples, the number of different possibilities (or hypotheses), would be doubly exponential in blocklength and the analysis of the optimal decoder becomes much more challenging than  classical (with constant number of distributions) identification problems. These differences in modeling the channel require a number of novel analytical tools.

\subsection{Contribution}
In this paper, we consider the SAS-MAC 
whose number of users increase exponentially with blocklength $n$. In characterizing the capacity of this model, we require its {\it global} probability of error to be vanishing. More specifically our contributions are as follows:
\begin{itemize}
\item
  We define a new massive identification paradigm in which we allow the number of sequences in a classical identification problem to increase exponentially with the sequence blocklength (or sample size). We find asymptotically matching upper and lower bounds on the probability of identification error for this problem. We use this result in our SAS-MAC model to recover the identity of the users.
  

\item 
We propose a new achievability scheme that supports strictly positive values of $(R, \alpha, \nu)$ for identical channels for the users.

\item 
We propose a new achievability scheme for the case that the channels of the users are chosen from a set of conditional distributions. The size of the set increases polynomially in the blocklength $n$. In this case, the channel statistics themselves can be used for user identification.

\item 
 We propose a new achievability scheme without imposing any restrictive assumptions on the users' channels. We show that strictly positive $(R, \alpha, \nu)$ are possible.

\item
We propose a novel converse bound for the capacity of  the SAS-MAC.

\item
  We show that for $\nu>\alpha$, not even reliable synchronization is  possible.

\end{itemize}

These results  were presented in parts in~\cite{shahi2016isit,shahi_asynch_ITW17,shahi_isit18}.

\subsection{Paper organization}
In Section~\ref{sec:identification} we introduce our massive identification model and  present a technical theorem (Theorem~\ref{thm:main identification}) that will be needed later on in the proof of Theorem~\ref{thm:achievable several channels}.  
In Section~\ref{sec:formulation} we introduce the SAS-MAC model and in Section~\ref{sec:results} we present our main results. 
More specifically, we introduce different achievability schemes for different scenarios and a converse technique to derive an upper bound on the capacity of the SAS-MAC. 
Finally, Section~\ref{sec:conclusion} concludes the paper.
Some proofs may be found in the Appendix.

\subsection{Notation}
\label{sec:notation}
Capital letters represent random variables that take on lower case letter values in calligraphic letter alphabets. 
The notation $a_n\doteq e^{nb}$ means $\lim_{n\to \infty}\frac{\log a_n}{n}=b$. 
We write $[M:N]$, where $ M, N\in \mathbb{ Z}, M\leq N$, to denote the set $\{M, M+1, \ldots, N\}$, and $[K]:=[1:K]$. We use $y_j^n :=[y_{j,1},..., y_{j,n}]$, and simply $y^n$ instead of $y_1^n$. 
The binary entropy function is defined by $h(p) := - p\log(p)-(1-p)\log(1-p)$.

\section{Massive Identification Problem }\label{sec:identification}
We first introduce notation specifically used in this  Section and then introduce our model and results.
\subsection{Notation}
When all elements of the random vector $X^n$ are generated i.i.d according to distribution $P$, we denote it as $X^n\simiid P$. We use $S_n$, where $\vert S_n\vert =n!$, to denote the set of all possible permutations of a set of $n$ elements. For a permutation $\sigma\in S_n$, $\sigma_i$ denotes the $i$-th element of the permutation. $\lfloor x\rfloor_r$ is used to denote the remainder of $x$ divided by $r$.  $K_k\left(a_1,\ldots, a_{\binom{k}{2}}\right)$ is the complete graph with $k$ nodes with edge index $ i\in [\binom{k}{2}]$ and edge weights $a_i, \ i\in [\binom{k}{2}]$. We may drop the edge argument and simply write $K_k$ when the edge specification is not needed. 
 A cycle $c$ of length $r$ in $K_k$ may be interchangeably defined by a vector of vertices as $c^{(v)}=\left[v_1,\ldots, v_r\right]$ or by a set of edges $c^{(e)}=\left\{a_1,\ldots, a_r \right\}$ where $a_i$ is the edge between $(v_i, v_{i+1}),\forall i\in [r-1]$ and $a_r$ is that between $(v_r, v_1)$. With this notation, $c^{(v)}(i)$ is then used to indicate the $i$-th vertex of the cycle $c$.
  $C^{(r)}_k$ is used to denote the set of all cycles of length $r$ in the complete graph $K_k\left(a_1,\ldots, a_{\binom{k}{2}}\right)$. The cycle gain, denoted by $G(c)$, for cycle $c=\left\{a_1,\ldots, a_r \right\}\in C^{(r)}_k$ is the product of the edge weights within the cycle $c$, i.e., $G(c) = \prod_{i=1}^ra_i, \ \forall a_i\in c$.  

The Bhatcharrya distance between $P_1$ and $P_2$ is denoted by $B(P_1,P_2):=\sum_{x\in \mathcal{X}}\sqrt{P_1(x)P_2(x)}$. 

 \subsection{Problem Formulation}
Let $P:=\{P_1, \ldots , P_{A}\}, P_i\in \mathcal{P}_{\mathcal{X}},\forall i\in [A]$ consist of $A$ distinct distributions and also let $\Sigma$  be uniformly distributed over $S_A$, the set of permutations of $A$ elements. In addition, assume that we have $A$ independent random vectors $\{X_1^n,X_2^n,\ldots, X_{A}^n\}$ of length $n$ each.  
For $\sigma$, a realization of $\Sigma$, assign the distribution $P_{\sigma_i}^n$ to the random vector $X_i^n, \forall i\in [A]$. 
After observing a sample $x^{nA}=[x_1^n,\ldots, x_A^n]$ of the random vector $X^{nA}=\left[X_1^n,\ldots, X_A^n \right]$, 
 we would like to identify $P_{\sigma_i}, \forall i\in[A]$.
 More specifically, we are interested in finding a permutation $\hsig: \mathcal{X}^{nA} \to S_A$ to indicate that $X_i^n\simiid P_{\hsig_i}, \ \forall i\in [A]$. Let $\Hsig = \hsig(X^{nA})$.

The average probability of error for the set of distributions $P$  is given by 
\begin{align}
P_e^{(n)}&=\PP\left[\hat{\Sigma}\neq \Sigma \right]\notag\\
&=\frac{1}{(A)!}\sum_{\sigma \in S_{A}}\PP\left[\hat{\Sigma} \neq \sigma\vert X_i^n \simiid P_{\sigma_i}, \forall i\in [A] \right]\notag\\
&=\PP\left[\hat{\Sigma} \neq \left(1,2,\ldots, A\right)\big\vert X_i^n \simiid P_{i}, \forall i\in [A] \right].\notag
\end{align}

We say that a set of distributions $P$ is identifiable if $\lim_{n\to \infty}P_e^{(n)}\to 0$. 
\subsection{Condition for Identifiability}
In Theorem~\ref{thm:main identification} we characterize the relation between the number of distributions and the pairwise distance of the distributions for reliable identification. Moreover, we introduce and use a novel graph theoretic technique in the proof of Theorem~\ref{thm:main identification} to analyze the optimal Maximum  Likelihood decoder.
\begin{theorem}\label{thm:main identification}
A sequence of distributions $P=\{P_1,\ldots,P_{A_n}\}$ is identifiable iff
\begin{align}
\lim_{n\to \infty}\sum_{\substack{1\leq i< j\leq A_n}}e^{-2nB(P_i,P_j)}=0.\label{eq:identifiability thm}
\end{align}
\end{theorem}
 
The rest of this section contains the proof. 
To prove Theorem~\ref{thm:main identification}, we provide  upper and lower bounds on the probability of error  in the following subsections.
\subsection{Upper bound on the probability of identification error}\label{subsec:achievability}
We use the optimal  Maximum Likelihood (ML) decoder, which minimizes the average probability of error, given by
\begin{align}
\hsig(x_1^n,\ldots, x_{A_n}^n):=\arg\max_{\sigma\in S_{A_n}}\sum_{i=1}^{A_n} \log \left(P_{\sigma_i}\left( x_i^n\right) \right),\label{eq:ml achieve}
\end{align}
where $P_{\sigma_i}\left( x_i^n\right)=\prod_{t=1}^nP_{\sigma_i}\left( x_{i,t}\right)$. 
The average probability of error associated with the ML decoder can also be written as
\begin{align}
P_e^{(n)}&=\PP\left[\Hsig \neq [A_n]\big\vert \widehat{H} \right]\notag
\\
&= \PP\left[\bigcup_{\hsig\neq [A_n]}\Hsig=\hsig\big\vert \widehat{H} \right]\notag
\\
&=\PP\left[\bigcup_{r=2}^{A_n} \bigcup_{\substack{\hsig:\\\left\{\sum_{i=1}^{A_n}\mathbbm{1}_{\{\hsig_i\neq i\}}=r \right\}}}  \Hsig =\hsig \big\vert \widehat{H} \right]\label{eq:2<r}
\\
&=\PP\Bigg[\bigcup_{r=2}^{A_n} \bigcup_{\substack{\hsig:\\\left\{\sum_{i=1}^{A_n}\mathbbm{1}_{\{\hsig_i\neq i\}}=r \right\}}} \sum_{i=1}^{A_n}\log\frac{P_{\hsig_i}\left(X_i^n \right)}{P_i\left(X_i^n \right)}\geq 0
  \big\vert \widehat{H}\Bigg],
\label{eq:error double cycle}
\end{align}
where $\widehat{H}:=\left\{X_i^n \simiid P_{i}, \forall i\in [A_n]\right\}$ and  where~\eqref{eq:2<r} is due to the requirement that each sequence is distributed according to a distinct distribution and hence the number of incorrect distributions ranges from $[2:A_n]$. 
In order to avoid considering the same set of error events multiple times, we incorporate a graph theoretic interpretation of $\left\{\sum_{i=1}^{A_n}\mathbbm{1}_{\{\Hsig_i\neq i\}}=r \right\}$ in~\eqref{eq:error double cycle} which is used to denote the fact that we have identified $r$ distributions incorrectly. Consider the two sequences $[i_1,\ldots, i_r]$ and $[\hsig_{i_1},\ldots, \hsig_{i_r}]$ for which we have 
\[\left\{\sum_{i=1}^{A_n}\mathbbm{1}_{\{\hsig_i\neq i\}}=\sum_{j=1}^{r}\mathbbm{1}_{\{\hsig_{i_j}\neq i_j\}}=r \right\}.\]
These two sequences in~\eqref{eq:error double cycle} in fact indicate the event that we have (incorrectly) identified $X_{i_j}^n\simiid P_{\hsig_{i_j}}$ instead of the (true) distribution  $X_{i_j}^n\simiid P_{{i_j}}, \forall j\in [r]$. 
 For a complete graph $K_{A_n}$, the set of edges between $\left((i_1,\hsig_{i_1}),\ldots, (i_r,\hsig_{i_r})\right)$ in $K_{A_n}$ would produce a single cycle of length $r$ or a set of disjoint cycles with total length $r$. However, we should note that in the latter case where the sequence of edges construct a set of (lets say of size $L$) disjoint cycles (each with some length $\tilde{r}_l$ for $\tilde{r}_l<r$ such that $\sum_{l=1}^L \tilde{r}_l=r$), then those cycles and their corresponding  sequences are already taken into account in the (union of) set of $\tilde{r}_l$ error events. 

As an example, assume $A_n=4$ and consider the error event
\[\log\frac{P_2(X_1^n)}{P_1(X_1^n)}+\log\frac{P_1(X_2^n)}{P_2(X_2^n)}+\log\frac{P_4(X_3^n)}{P_3(X_3^n)}+\log\frac{P_3(X_4^n)}{P_4(X_4^n)}\geq 0,\]
which corresponds to the (error) event of choosing $[\hsig_1, \hsig_2,\hsig_3,\hsig_4]=[2,1, 4, 3]$ over $[1,2,3,4]$ with $r=4$ errors. In the graph representation, this gives two cycles of length $2$ each, which correspond to
\begin{align*}
\left\{\log\frac{P_2(X_1^n)}{P_1(X_1^n)}+\log\frac{P_1(X_2^n)}{P_2(X_2^n)}\geq 0 \ \right\} \cap
\left\{\log\frac{P_4(X_3^n)}{P_3(X_3^n)}+\log\frac{P_3(X_4^n)}{P_4(X_4^n)}\geq 0\right\},
\end{align*} 
and are already accounted for in the events 
\[\left\{[\hsig_1, \hsig_2,\hsig_3,\hsig_4]=[2,1, 3, 4]\right\}\cup \left\{[\hsig_1, \hsig_2,\hsig_3,\hsig_4]=[1,2, 4, 3]\right\}\]
 with $r=2$.
 
 As the result, in order to avoid double counting, in evaluating~\eqref{eq:error double cycle} for each $r$ we should only consider the sets of sequences which produce a {\it single} cycle of length $r$.
 
 Before proceeding further, we define the edge weights for a complete weighted graph \[K_{A_n} (a_{(1,2)},\ldots a_{(K_{n},1)}).\] In particular, we define $a_{(i,j)}:=e^{-nB(P_i,P_j)}$ to be the edge weight between vertices $(i,j)$  in the complete graph $K_{A_n}$ shown in Fig.~\ref{fig:graph cycles}. 
\begin{figure}[htbp]
\centering
\includegraphics[width = .5 \textwidth]{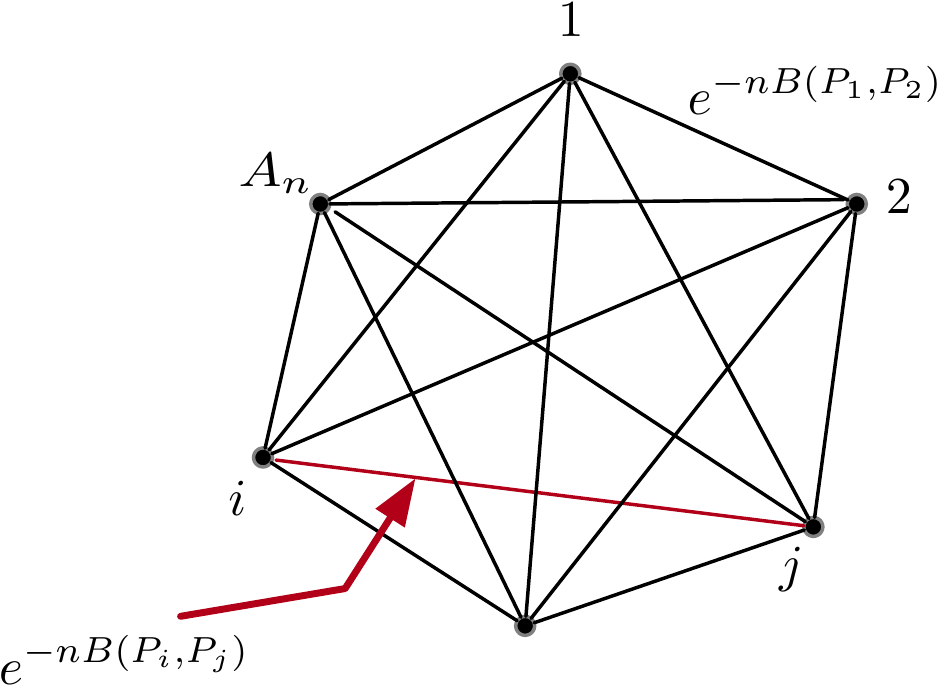}
\caption{Complete graph $K_{A_n}$ with edge weight $e^{-nB(P_i,P_j)}$ for every pair of vertices $i\neq j\in [K_n]$.}
\label{fig:graph cycles}
\end{figure}
 
  Hence, we can upper bound the probability of error in~\eqref{eq:error double cycle} as 
\begin{align}
P_e^{(n)}&\leq\sum_{r=2}^{A_n} 
\sum_{\substack{c\in C^{(r)}_{A_n}}} \PP\left[\sum_{i=1}^r\log \frac{P_{\lfloor c^{(v)}(i+1)\rfloor_{r}}\left(X^n_{c^{(v)}(i)}\right)}{P_{c^{(v)}(i)}\left(X^n_{c^{(v)}(i)}\right)}\geq 0\vert \widehat{H}\right]\notag\\
&\leq \sum_{r=2}^{A_n} 
\sum_{\substack{c\in C^{(r)}_{A_n}}} e^{-n\sum_{i=1}^r B(P_{c^{(v)}(i)}, P_{c^{(v)}(\lfloor i+1\rfloor_{r})})}\label{eq:ml achievability}
\\
&=\sum_{r=2}^{A_n} \sum_{c\in C_{A_n}^{(r)}} G(c)\label{eq:achieve graph},
\end{align}
where $r$ enumerates the number of incorrect matchings 
and where $c(i)$ is the $i$-th vertex in the cycle $c$. In~\eqref{eq:achieve graph}, we have leveraged the fact that $e^{-nB(P_i,P_j)}$ is the edge weight between vertices $(i,j)$  in the complete graph $K_{A_n}$ and hence $ G(c)=e^{-n\sum_{i=1}^r B(P_{c^{(v)}(i)}, P_{c^{(v)}(\lfloor i+1\rfloor_{r})})}$ is the  gain of cycle $c$. 
The inequality in~\eqref{eq:ml achievability} is by
\begin{align}
 &
 \PP\left[\sum_{i=1}^r\log \frac{P_{\lfloor c^{(v)}(i+1)\rfloor_{r}}\left(X^n_{c^{(v)}(i)}\right)}{P_{c^{(v)}(i)}\left(X^n_{c^{(v)}(i)}\right)}\geq 0\vert \widehat{H}\right]\notag
\\&\leq \exp\left\{n\inf_t \log \mathbb{E}\left[ \prod_{i=1}^r \left(\frac{P_{ c^{(v)}(\lfloor i+1\rfloor_{r})}\left(X^n_{c(i)}\right)}{P_{c^{(v)}(i)}\left(X^n_{c(i)}\right)}\right)^t \right]\right\}\notag
\\&\leq  \exp\left\{n\sum_{i=1}^r \log \mathbb{E}\left[  \left(\frac{P_{c^{(v)}(\lfloor i+1\rfloor_{r})}\left(X^n_{c(i)}\right)}{P_{c^{(v)}(i)}\left(X^n_{c(i)}\right)}\right)^{1/2} \right]\right\}\label{eq:t = 1/2}
\\&=  \exp\left\{-n\sum_{i=1}^r B(P_{c^{(v)}(i)}, P_{ c^{(v)}(\lfloor i+1\rfloor_{r})})\right\}.\notag
\end{align}

The fact that we used $t=1/2$ in~\eqref{eq:t = 1/2} instead of finding the exact optimizing $t$, comes from  the fact that $t=1/2$ is the optimal choice for $r=2$ and as we will see later, the rest of the error events are dominated by the set  of only $2$ incorrectly identified distributions. This can be seen as follows for $X_1^n\simiid P_1, X_2^n\simiid P_2$
\begin{align}
&\PP\left[\log\frac{P_1(X^n_2)}{P_2(X^n_2)}+\log\frac{P_2(X_1^n)}{P_1(X_1^n)}\geq 0 \right]\notag
\\
&=\sum_{\substack{\hat{P}_1,\hat{P_2}:\\\sum_{x\in \mathcal{X}}\hat{P}_1(x)\log \frac{P_2(x)}{P_1(x)}+\\\hat{P_2}(y)\log\frac{P_1(x)}{P_2(x)}\geq 0}} \exp\left\{ nD\left(\hat{P}_1\parallel P_1\right)\!-\!nD\left(\hat{P}_2\parallel P_2 \right)\!\right\}\notag
\\&\doteq e^{-nD\left(\tilde{P} \parallel P_1\right)-nD\left(\tilde{P} \parallel P_2 \right)}=e^{-2nB(P_1, P_2)},\label{eq:lagrange}
\end{align}
where $\tilde{P}$ in the first equality in~\eqref{eq:lagrange}, by using the Lagrangian method, can be shown to be equal to $\tilde{P}(x)=\frac{\sqrt{P_1(x)P_2(x)}}{\sum_{x'}\sqrt{P_1(x')P_2(x')}}$ and subsequently the second inequality in~\eqref{eq:lagrange} is proved.

In order to calculate the expression in~\eqref{eq:achieve graph}, we use the following graph theoretic Lemma, the proof of which is given in Appendix~\ref{app:lemma graph proof}.

\begin{lemma}\label{lemma:graph}
In a complete graph $K_{k}\left(a_1,\ldots, a_{n_k}\right)$ and for the set of cycles of length $r$,  $\mathcal{C}_{k}^{(r)}=\{c_1,\ldots c_{N_{r,k}}\}$, we have
\begin{align}
\frac{1}{N_{r,k}}\left( G(c_1)+\ldots G(c_{N_{r,k}})\right)& \leq \left(\frac{a_1^2+\ldots+ a_{n_k}^2}{n_k}\right)^{\frac{r}{2}}\notag
\end{align}
where $N_{r,k}, n_k$ are the number of cycles of length $r$ and the number of edges in the complete graph $K_k$, respectively.
\end{lemma}
By Lemma~\ref{lemma:graph} and~\eqref{eq:achieve graph} we prove in Appendix~\ref{app:ident ach} that 
 the upper bound on the  probability of error $P_e^{(n)}$  goes to zero if 
\begin{align}
\lim_{n\to \infty}\sum_{\substack{1\leq i<j\leq A_n}}e^{-2nB(P_i,P_j)}=0.\label{eq:ident ach}
\end{align}
As a result of Lemma~\ref{lemma:graph}, it can be seen from~\eqref{eq:nN} that the sum of probabilities that $r\geq 3$ distributions are incorrectly identified is dominated by the probability that only $r=2$ distributions are incorrectly identified. This shows that the most probable error event is indeed an error event with two wrong distributions.

\subsection{Lower bound on the probability of  identifiability error}
For our converse, we use the optimal ML decoder, and as a lower bound to the probability of error in~\eqref{eq:error double cycle}, we only consider the set of error events with only two incorrect distributions, i.e., the set of events with $r=2$. In this case we have
\begin{align}
P_e^{(n)}&\geq \PP\left[ \bigcup_{\substack{1\leq i<j\leq A_n}} \log\frac{P_i(X^n_j)}{P_j(X^n_j)}+\log\frac{P_j(X_i^n)}{P_i(X_i^n)}\geq 0\vert \widehat{H}\right]\notag
\\&\geq \frac{\left(\sum_{\substack{1\leq i<j\leq A_n}}  \PP\left[\xi_{i,j}\right]\right)^2}{\sum_{\substack{(i,j),(j,k)\\(i,j)\neq(l,k)\\i\neq j, l\neq k}}\PP[\xi_{i,j},\xi_{k,l}]},
\label{eq:conv lb}
\end{align}
where~\eqref{eq:conv lb} is by~\cite{chung1952application}
 and where 
\begin{align}
&\xi_{i,j}:= \left\{\log\frac{P_i}{P_j}(X^n_j)+\log\frac{P_j}{P_i}(X_i^n)\geq 0\vert \widehat{H}\right\}.
\label{eq:1/2 optimal}
\end{align}
We prove in Appendix~\ref{app:ident conv} that a lower bound on $P_e^{(n)}$ is given by 
\begin{align}
P_e^{(n)}\geq &\frac{\left(\sum_{1\leq i<j\leq A_n}e^{-2nB(P_i,P_j)}\right)^2}{\sum\limits_{i,j,k}e^{-nB(P_i,P_j)-nB(P_i,P_k)-nB(P_k,P_j)}\!+\!\left(\sum\limits_{i,j}e^{-2nB(P_i,P_j)}\!\right)^2}\label{eq:ident conv}
\\& \geq \frac{\left(\sum_{i,j}e^{-2nB(P_i,P_j)}\right)^2}{8\left(\sum\limits_{1\leq i<j\leq A_n}\!\!\!e^{-2nB(P_i,P_j)}\right)^{\!\frac{3}{2}}\!\!\!+\left(\sum\limits_{1\leq i<j\leq A_n}e^{-2nB(P_i,P_j)}\right)^2}
\label{eq:conv simplified}\\
&=\frac{\sqrt{\sum_{1\leq i<j\leq A_n}e^{-2nB(P_i,P_j)}}}{8+\sqrt{\sum_{1\leq i<j\leq A_n}e^{-2nB(P_i,P_j)}}}, \label{eq:conv final lb}
\end{align}
where~\eqref{eq:conv simplified} is by Lemma~\ref{lemma:graph}. 
As it can be seen from~\eqref{eq:conv final lb}, if  $\lim_{n\to \infty }\sum_{\substack{1\leq i<j\leq A_n}}e^{-2nB(P_i,P_j)}\neq 0$, the probability of error is bounded away from zero. As a result, we have to have 
\[\lim_{n\to \infty }\sum_{\substack{1\leq i<j\leq A_n}}e^{-2nB(P_i,P_j)}=0,\] which also matches our upper bound on the probability of error  in~\eqref{eq:achieve final ub}.

\begin{remark} 
  As it is clear from the result of Theorem~\ref{thm:main identification},  when $A_n$ is a constant or  grows polynomially with $n$, the sequence of distributions in $P$ are always identifiable and the probability of error in the identification problem decays to zero as the blocklength $n$ goes to infinity. The interesting aspect of Theorem~\ref{thm:main identification} is in the regime that $A_n$ increases exponentially with the blocklength.
  \end{remark} 
  
  Having proved the criterion for identifiability of a massive number of distributions in Theorem~\ref{thm:main identification}, we move on to the SAS-MAC problem. We use the result of Theorem~\ref{thm:main identification}   to identify the massive number of users  by their induced probability distribution at the receiver.
\section{SAS-MAC problem}
\label{sec:formulation}
We first introduce the special notation used in the SAS-MAC  and then formally define the problem.
\subsection{Special Notation}
A stochastic kernel / transition probability / channel from $\mathcal{X}$ to $\mathcal{Y}$ is denoted by $Q(y|x), \forall (x,y)\in \mathcal{X}\times\mathcal{Y}$, and the output marginal distribution induced by $P\in \mathcal{P}_{\mathcal{X}}$ through the channel $Q$ as 
\begin{align}
[PQ](y) := \sum_{x\in\mathcal{X}} P(x)Q(y|x), \forall  y\in\mathcal{Y},
\label{eq:defYmarginal}
\end{align} 
where $\mathcal{P}_{\mathcal{X}}$ is the space of all distributions on $\mathcal{X}$. 
We define the shorthand notation
\begin{align}
Q_{x}(y) &:=Q(y|x), \forall  y\in\mathcal{Y}.
\label{eq:defQx(y)}
\end{align}

For a MAC channel $Q(y|x_1,\ldots,x_K), \forall (x_1,\ldots,x_K,y)\in \mathcal{X}_1\times\ldots\times\mathcal{X}_K\times\mathcal{Y}$, we define the shorthand notation 
\begin{align}
Q_S\left( y\vert x_S\right) &:=Q(y\vert x_S, \star_{S^c}), \forall S\subseteq [K],
\label{eq:defQS(y|xS)}
\end{align}
to indicate that users indexed by $S$ transmit $x_i, \forall i \in S$, and users indexed by $S^c$ transmit their respective idle symbol $\star_j, \forall j \in S^c=[K]\backslash S$.
 When $\vert S\vert =1$, we use 
 \[Q_{i}(y\vert x_i):= Q_{\{i\}}(y\vert x_i)=Q(y\vert \star_1,\ldots,\star_{i-1}, x_i, \star_{i+1},\ldots, \star_K),\]
  and when $\vert S\vert =0$, we use \[Q_\star(y):=Q(y\vert \star_1,\ldots \star_K).\]

The empirical distribution of a sequence $x^n$ is 
\begin{align}
\widehat{P}_{x^n}(a)
:=\frac{1}{n}\mathcal{N}(a\vert x^n)=\frac{1}{n}\sum_{i=1}^n \mathbbm{1}_{\{x_i=a\}}, 
\forall a \in \mathcal{X},
\label{eq:def empirical dist}
\end{align}
where $\mathcal{N}(a\vert x^n)$ denotes the number of occurrences of letter $a\in \mathcal{X}$ in the sequence $x^n$; 
when using~\eqref{eq:def empirical dist} the target sequence $x^n$ is usually clear from the context so we may drop  the subscript $x^n$ in $\widehat{P}_{x^n}(\cdot)$.
The $P$-type set and the $V$-shell of the sequence $x^n$ are defined, respectively, as 
\begin{align}
T(P)&:=\left\{x^n: \mathcal{N}(a\vert x^n)=nP(a),\forall a\in \mathcal{X} \right\},
\label{eq:def Ptype}
\\
T_V(x^n)&:=\left\{y^n\!:\! \frac{\mathcal{N}\left(a,b\vert x^n, y^n \right)}{\mathcal{N}(a\vert x^n)}=V(b|a),\forall (a, b)\in (\mathcal{X},\mathcal{Y})\right\},
\label{eq:def Vshell}
\end{align}
where $\mathcal{N}(a,b\vert x^n, y^n) =\sum_{i=1}^n \mathbbm{1}_{\left\{\substack{x_i=a\\y_i=b} \right\}}$ is the number of joint occurrences of $(a,b)$ in the pair of sequences $(x^n, y^n)$.

We use 
$D(P_1 \parallel P_2)$ to denote the Kullback Leibler divergence between distribution $P_1$ and $P_2$,
and $D(Q_1 \parallel Q_2|P) :=\sum_{x,y\in\mathcal{X}\times \mathcal{Y}} P(x)Q_1(y|x) \log \frac{Q_1(y|x)}{Q_2(y|x)}$ for the conditional Kullback Leibler divergence.   
We let $I(P,Q) = D(Q \parallel [PQ]|P)$ denote the mutual information between random variable {$(X,Y)$ with joint distribution $P_{X,Y}(x,y)=P(x)Q(y|x)$.

\subsection{ SAS-MAC Problem Formulation}
Let
 $M$ be the number of messages,
 $A$ be the number of blocks, and
 $K$ be the number of users.
An $(M, A,K,n, \epsilon)$ code for the SAS-MAC 
consists of:
\begin{itemize}

\item 
A message set $[M]$, for each user $i \in [K]$, from which messages are chosen uniformly at random and are independent across users.

\item
An encoding function $f_i: [M]\to \mathcal{X}^n$, for each user $i\in [K]$. We define 
\begin{align}
x_i^n(m):= f_i(m).
\end{align} 
Each user $i\in[K]$ choses a message $m_i\in [M]$ and a block index $t_i \in [A]$, both uniformly at random. It then transmits $[\star_i^{n(t_i-1)} \ x_i^n(m_i) \ \star_i^{n(A-t_i)}]$, where $\star_i\in \mathcal{X}$ is the designated `idle' symbol for user $i$. 

\item
A destination decoding function 
\begin{align}
\left((\widehat{t}_1, \widehat{m}_1),\ldots, (\widehat{t}_K, \widehat{m}_K)\right)= g(\mathcal{Y}^{nA}),
\end{align}
such that its associated probability of error, $P_e^{(n)}$, satisfies $P_e^{(n)} \leq \epsilon$ where 
\begin{align}
P_e^{(n)}:= \frac{1}{\left( A M\right)^K}\sum_{(t_1, m_1),\ldots, (t_K,m_K)} \PP\left[g(Y^n)\neq \left((t_1, m_1),\ldots, (t_K,m_K) \right) \vert H_{\left((t_1, m_1),\ldots, (t_K,m_K) \right)}\right],
\end{align}
where the hypothesis that user $i\in [K]$ has chosen message $m_i \in [M]$ and block $t_i \in [A]$ is denoted by $ H_{\left((t_1, m_1),\ldots, (t_K,m_K) \right)}$.
 
\end{itemize}

A tuple $(R, \alpha, \nu)$ is said to be achievable if there exists a sequence of codes $\left(e^{nR},e^{n\alpha}, e^{n\nu}, n,\epsilon_n\right)$ with $\lim_{n\to \infty }\epsilon_n = 0$. The capacity region of the  SAS-MAC 
at asynchronous exponent $\alpha$, occupancy exponent $\nu$ and rate $R$, is the closure of all possible achievable $(R, \alpha, \nu)$ triplets.

\section{Main results for SAS-MAC}~\label{sec:results}
In this Section we first introduce an achievable region for the case that different users have identical channels (in Theorem~\ref{thm:same channel achievability}).  We then move on to the more general case where the users' channels belong to a set of conditional probability distributions of polynomial size in $n$ (in Theorem~\ref{thm:achievable several channels}). In this case, we use the output statistics to distinguish and identify the users.  We then remove all  restrictions on the users' channels and derive an achievability bound on the capacity of the SAS-MAC (in Theorem~\ref{thm:ach general}). After that, we propose a converse bound on the capacity of general SAS-MAC (in Theorem~\ref{thm:converse sas-mac}).  We then provide a converse bound on the  number of users (in Theorem~\ref{thm:conv nu}).

\subsection{Users with Identical Channels}
\label{subsec:Users with Identical Channels}
The following theorem is an achievable region for the 
SAS-MAC for the case that different users have identical channels toward the base station when they are the sole active user. In this scenario, users' identification and decoding can be merged together.

\begin{theorem}
\label{thm:same channel achievability}
For a SAS-MAC with asynchronous exponent $\alpha$, occupancy exponent $\nu$ and rate $R$, 
assume that $Q_{\{i\}}(y\vert x) 
 = Q(y\vert x)$ (recall definition~\eqref{eq:defQS(y|xS)}) for all users.
Then, the following $(R, \alpha, \nu)$ region is achievable
\begin{align}
\bigcup_{\substack{P\in \mathcal{P}_{\mathcal{X}}\\\lambda \in [0,1]}}
\left\{\begin{matrix}
\nu&< \frac{\alpha}{2}\\
\nu&< D(Q_{\lambda} \parallel Q|P)\\
\alpha+R+\nu&<D(Q_{\lambda} \parallel Q_\star|P)\\
R+\nu&<I(P,Q)
\end{matrix}\right\},
\label{eq:ach same channel}
\end{align}
where 
\begin{align}
Q_\lambda(y\vert x) := \frac{Q(y\vert x)^\lambda Q_\star(y)^{1-\lambda}}{\sum_{y^\prime\in \mathcal{Y}}Q(y^\prime\vert x)^\lambda Q_\star(y^\prime)^{1-\lambda}}, \ \forall (x,y)\in \mathcal{X}\times\mathcal{Y}.
\end{align}
\end{theorem}

\begin{proof}
Before starting the proof, we note that for $\nu<\frac{\alpha}{2}$ (first bound in~\eqref{eq:ach same channel}), with probability approaching   one as the blocklength $n$ grows to infinity, the users transmit in distinct blocks. Hence, in analyzing the joint probability of error of our achievability scheme, we can safely condition on the hypothesis that users do not collide. The probability of error given the hypothesis that collision has occurred, which may be large, is then multiplied by the probability of collision and hence is vanishing as the blocklength goes to infinity, regardless of the achievable scheme.  The probability of error for this two-stage decoder can be decomposed as
\begin{align}
\PP\left[ \text{Error}\right] 
  &=\PP\left[\text{Synchronization error} \right]
\\&+\PP\left[\text{Decoding error} \vert \text{No synchronization error}\right].
\end{align}

\paragraph*{Codebook generation} 
Let
$K_n = e^{n\nu}$ be the number of users,
$A_n = e^{n\alpha}$ be the number of blocks, and
$M_n = e^{nR}$ be the number of messages.
Each user $i\in [K_n]$ generates a constant composition codebook with composition $P$ by drawing each message's codeword uniformly and independently at random from the $P$-type set $T(P)$ (recall definition in~\eqref{eq:def Ptype}). The codeword of user $i\in [K_n]$ for message $m\in [M_n]$ is denoted as $x_i^n(m)$.

\paragraph*{Probability of error analysis} 
A two-stage decoder is used, to first synchronize  and then decode (which also identifies the users' identities) the users' messages. We now introduce the two stages and bound the probability of error for each stage.

{\bf Synchronization step.} 
We perform a sequential likelihood test as follows.
Fix a threshold
\begin{align}
T\in [-D(Q_\star\parallel Q\vert P),D(Q\parallel Q_\star\vert P)].
\label{eq:def T ThmSameCh}
\end{align} 
For each block $j\in [A_n]$ if there exists any message $m\in [M_n]$ for any user $i\in [K_n]$ such that 
\begin{align}
L(x_i^n(m), y_j^n):=\frac{1}{n}\log \frac{Q(y_j^n\vert x_i^n(m))}{Q_\star(y_j^n)}\geq T,
\label{eq:def LLR ThmSameCh}
\end{align}
then declare that block $j$ is an `active' block, and an `idle' block otherwise. 
Let 
\begin{align}
H^{(1)} := 
 H_{\left((1,1),(2,1),\ldots, (K_n,1) \right)},
\label{eq:def Hyp ThmSameCh}
\end{align}
be the hypothesis that user $i\in[K_n]$ is active in block $i$ and sends message $m_i=1$.
The average probability of synchronization error, averaged over the different hypotheses, is upper bounded by
\begin{align}
& \PP\left[\text{Synchronization error}\right]
 =\PP\left[\text{Synchronization error} \vert H^{(1)} \right]
\label{eq:symmetry hyp ThmSameCh}\\
&\leq
  \sum_{j=1}^{K_n} \PP\left[\bigcap_{i=1}^{K_n} \bigcap_{m=1}^{M_n} L(x_i^n(m),Y_j^n)< T   \vert  H^{(1)}\right]
+ \sum_{j=K_n+1}^{A_n} \PP\left[\bigcup_{i=1}^{K_n} \bigcup_{m=1}^{M_n} L(x_i^n(m), Y_j^n)\geq T  \vert  H^{(1)}\right]
\\
&\leq
  \sum_{j=1}^{K_n}\PP\left[L(x_j^n(1),Y_j^n)< T   \vert  H^{(1)}\right]
 +\sum_{j=K_n+1}^{A_n}\sum_{i=1}^{K_n} \sum_{m=1}^{M_n}  \PP\left[L(x_i^n(m), Y_j^n)\geq T  \vert  H^{(1)}\right]
\\ 
&\leq e^{n\nu}e^{-nD\left( Q_\lambda \parallel Q\vert P\right)}+e^{n(\alpha+\nu+R)}e^{-nD\left(Q_\lambda\parallel Q_\star\vert P \right)}
\label{eq:usual lagrangian},
\end{align}
where~\eqref{eq:symmetry hyp ThmSameCh} is by the symmetry of different hypothesis and~\eqref{eq:usual lagrangian} can be derived as in~\cite[Chapter 11]{coverbook}.
The upper bound on the probability of error for the synchronization error in~\eqref{eq:usual lagrangian} vanishes as $n$ goes to infinity if the second and third bound in~\eqref{eq:ach same channel} hold. 

{\bf Decoding stage.} 
In this stage, by conditioning on no synchronization error, we have a superblock of length $nK_n$, for which we have to distinguish between $K_n! (M_n)^{K_n}\doteq e^{nK_n(R+\nu)}$ different messages. We note that all the codewords for this superblock also have constant composition $P$ (since they are formed by 
 the concatenation of constant composition codewords). We can hence use a Maximum Likelihood (ML) decoder for random constant composition codes, introduced and analyzed in~\cite{moulin2012log}, on the super-block of length $nK_n$ to distinguish among $e^{nK_n(R+\nu)}$ different messages with vanishing probability of error if $R+\nu<I(P,Q)$. This retrieves the last bound in~\eqref{eq:ach same channel}.
\end{proof}
 
\subsection{Users with Different Choice of Channels}
\label{subsec:Users with Different Channels}
We now move on to a more general case in which we remove the restriction that the channels of all users are the same. Theorem~\ref{thm:achievable several channels} finds an achievable region when we allow the users' channels to be chosen from a set of conditional distributions of polynomial size in the blocklength.

\begin{theorem}
\label{thm:achievable several channels}
For a SAS-MAC with asynchronous exponent $\alpha$, occupancy exponent $\nu$ and rate $R$,
assume that $Q_i(y\vert x)=W_{c(i)}(y\vert x)$ is the channel for user $i\in [K_n]$, for some $c(i)\in[S_n]$ where $S_n=\text{poly}(n)$.
Then, the following region is achievable
\begin{align}
\bigcup_{n\geq 1}\bigcup_{\substack{P_j\in \mathcal{P}_{\mathcal{X}}\\\lambda_j \in [0,1]}}\bigcap_{j\in [S_n]}\left\{
\begin{matrix}
\nu_j &< \frac{\alpha}{2} \\
\nu_j &<D([P_jW_j]_{\lambda_j} \parallel [P_jW_j]) \\
\alpha&<D([P_jW_j]_{\lambda_j} \parallel Q_{\star})\\
R+\nu_j&<I(P_j,W_j)\\
\end{matrix}\right\},
\label{eq:ach different channel}
\end{align}
where
\begin{align}
\nu_j &:= \frac{1}{n}\log(\mathcal{N}_j), 
\\
\mathcal{N}_j &:=\sum_{i=1}^{K_n} \mathbbm{1}_{\{Q_i =W_j\}} : \sum_{j=1}^{S_n} \mathcal{N}_j =  K_n. 
\end{align}
\end{theorem}

\begin{proof} 
Before starting the proof, we should note that with similar arguments as the ones in Theorem~\ref{thm:same channel achievability}, by imposing the first bound in~\eqref{eq:ach different channel}, different users transmit in distinct blocks with a probability which goes to one as blocklength goes to infinity; thus we can assume no user collision in the following.
We now propose a three-stage achievability scheme. 
The three stages perform the task of synchronization, identification and decoding, respectively. 
The joint probability of error for this three-stage achievable scheme can be decomposed as
\begin{align}
\PP\left[ \text{Error}\right]
  &=\PP\left[\text{Synchronization error} \right]
\\&+\PP\left[\text{Identification error} \vert \text{No synchronization error}\right]
\\&+\PP\left[\text{Decoding error}\vert\text{No synchronization and No identification error} \right].
\end{align}

\paragraph*{Codebook generation}
Let
$K_n = e^{n\nu}$ be the number of users,
$A_n = e^{n\alpha}$ be the number of blocks,
$M_n = e^{nR}$ be the number of messages, and
$S_n=\text{poly}(n)$ be the number of channels.
Each user $i\in [K_n]$ generates a random i.i.d codebook according to distribution $P_{c(i)}$ where the index $c(i)\in[S_n]$ is chosen based on the channel $Q_i = W_{c(i)}$. For each user $i\in [K_n]$, the codeword for each message $m\in [M_n]$ is denoted as $x_i^n(m)$.

\paragraph*{Probability of error analysis} 
A three-stage decoder is used. We now introduce the three stages and bound the probability of error for each stage.

{\bf Synchronization step.} 
We perform a sequential likelihood ratio test for synchronization  as follows. 
 Recall $Q_i(\cdot \vert \cdot)=W_{c(i)}(\cdot \vert \cdot)$ for all user $i\in [K_n]$. 
Fix thresholds 
\begin{align}
T_{c(i)}\in \left[-D\left(Q_{\star} \parallel [P_{c(i)}W_{c(i)}] \right), D\left([P_{c(i)}W_{c(i)}] \parallel Q_{\star} \right) \right], \ i\in [K_n].
\label{eq:def T ThmDifferentCh}
\end{align} 
For each block $j\in [A_n]$ if there exists any user $i\in [K_n]$ such that 
\begin{align}
L_i(y_j^n) := \frac{1}{n}\log \frac{[P_{c(i)}W_{c(i)}](y_j^n)}{Q_{\star}(y_j^n)}\geq  T_{c(i)},
\label{eq:def LLR ThmDifferentCh}
\end{align}
then declare that block $j$ is an `active' block. Else, declare that block $j$ is an `idle' block. Note that  were able to calculate the probabilities of error corresponding to~\eqref{eq:def LLR ThmSameCh} by leveraging the constant composition construction of codewords in Theorem~\ref{thm:same channel achievability}. In here, we can leverage the i.i.d. constructure of the codewords and calculate the probability of error corresponding to~\eqref{eq:def LLR ThmDifferentCh}.

We now find an upper bound on the average probability of error for this scheme over different hypotheses. Before doing so, we should note that by the symmetry of different hypotheses, the average probability of error over different hypothesis is equal to probability of error given the hypothesis that user $i\in [K_n]$ transmits in block $i$; this hypothesis is denoted by 
\begin{align}
H^{(2)} := 
H_{\left((1,.),(2,.),\ldots, (K_n,.) \right)} ,
\label{eq:def Hyp ThmDifferentCh}
\end{align}
 where a dot, as in $(.)$, is used instead of specifying the messages to emphasize that the decoder finds the location of the users, irrespective of their transmitted messages.

The average probability of synchronization error, averaged over the different hypotheses, is upper bounded by
\begin{align}
& \PP\left[\text{Synchronization error}\right]
 =\PP\left[\text{Synchronization error}  \vert H^{(2)} \right]
\label{eq:symmetry hyp ThmDifferentCh}\\
&\leq
  \PP\left[\bigcup_{i=1}^{K_n}   L_i( Y_i^n) <  T_{c(i)}  \vert H^{(2)} \right]
+ \PP\left[\bigcup_{j=K_n+1}^{A_n}\bigcup_{i=1}^{K_n}  L_i(Y_j^n)\geq   T_{c(i)}  \vert H^{(2)} \right]
\\
&\leq 
 \sum_{i=1}^{K_n} \PP\left[L_i(Y_i^n) <    T_{c(i)}  \vert H^{(2)} \right]
 +(A_n-K_n) 
 \PP\left[\bigcup_{i=1}^{K_n} L_i( Y^n)\geq  T_{c(i)}   \vert H^{(2)} \right]\\
&\leq \sum_{i=1}^{K_n} e^{-nD\left([P_{c(i)}W_{c(i)}]_{\lambda_i} \parallel [P_{c(i)}W_{c(i)}]\right)} 
+e^{n\alpha} \sum_{j=1}^{S_n}e^{-n D\left([P_{j}W_j]_{\lambda_j} \parallel  Q_{\star} \right)}\\
&= \sum_{j=1}^{S_n}\mathcal{N}_j e^{-nD\left([P_{j}W_{j}]_{\lambda_j} \parallel [P_{j}W_{j}]\right)}
+e^{n\alpha} e^{-nD\left([P_{j}W_j]_{\lambda_j} \parallel  Q_{\star} \right)},
\end{align}
where
\begin{align}
[P_jQ_j]_{\lambda}(y):=\frac{\left( [P_jQ_j](y)\right)^{\lambda} \left(Q_{\star}(y)\right)^{1-\lambda}}{\sum_{y^\prime\in \mathcal{Y}}\left( [P_jQ_j](y^\prime)\right)^{\lambda} \left(Q_{\star}(y^\prime)\right)^{1-\lambda}}.
\end{align}
The probability of error in this stage will decay to zero if for all $j\in [S_n]$
\begin{align}
\nu_j := \frac{1}{n}\log(\mathcal{N}_j)  &<D\left([P_{j}W_{j}]_{\lambda_j} \parallel [P_{j}W_{j}]\right),\\
\alpha&< D\left([P_{j}W_j]_{\lambda_j} \parallel Q_{\star} \right).
\end{align}
This retrieves the  second and third bounds in~\eqref{eq:ach different channel}.

{\bf Identification step.}
Having found the location of the `active' blocks, we move on to the second stage of the achievability scheme to identify which user is active in which block.
We note that, by the random codebook generation and the memoryless property of the channel, the output of the block occupied by user $i\in [K_n]$ is i.i.d distributed according to the marginal distribution
\begin{align}
[P_{c(i)} Q_{c(i)}](y):=\sum_{x\in \mathcal{X}}P_{c(i)}(x) Q_{c(i)}(y\vert x). 
\end{align}
We leverage this property and customize the result in Theorem~\ref{thm:main identification} to identify the different distributions of the different users. Note that at this point, we only distinguish the users with different channels from one another. In Theorem~\ref{thm:main identification}, it was assumed that all the distributions are distinct, while in here, our distributions are not necessarily distinct. The only modification that is needed in order to use the result of Theorem~\ref{thm:main identification} is  as follows. We need to consider a graph in which the edge between every two similar distributions have edge weights equal to zero (as opposed to $e^{B(P,P)}= e^{0}=1$). By doing so, we can easily conclude that the probability of identification error in our problem using an ML decoder goes to zero as blocklength $n$ goes to infinity since
\begin{align}
\PP[\text{Identification error}\vert \text{No synchronization error}]\leq \sum_{1\leq i<j\leq S_n}e^{-2nB\left([P_iQ_i], [P_jQ_j] \right)}\to 0, \label{eq:identification step error}
\end{align}
and since $S_n=\text{poly}(n)$ by assumption.

{\bf Decoding stage.} 
After finding the permutation of users in the active blocks, we can go ahead with the third stage of the achievable scheme to find the transmitted messages of the users.
In this stage, we can group the users who have similar channel $Q_j$ to get superblocks of length $n\mathcal{N}_j, j \in[S_n]$. For each superblock, we have to distinguish $(M_n)^{\mathcal{N}_j} (\mathcal{N}_j)^{\mathcal{N}_j}\approx e^{n\mathcal{N}_j(R+\nu_j)}$ different message permutations. By using a typicality decoder, we conclude that the probability of decoding error for each superblock will go  exponentially fast in blocklength to zero if
\begin{align}
R+\nu_j<I(P_j,W_j), \forall j\in [S_n].
\end{align}
This retrieves the last bound in~\eqref{eq:ach different channel} and concludes the proof.
\end{proof}

\begin{remark}\label{rmk:iid vs constant comp}
The achievability proof of Theorem~\ref{thm:same channel achievability} relies on constant composition codes whereas the achievability proof of Theorem~\ref{thm:achievable several channels} relies on i.i.d. codebooks. The reason for these restrictions is that in~\ref{thm:achievable several channels} we also need to distinguish different users and in order to use the result of~\cite{shahi_isit18}, we focused our attention on i.i.d. codebooks.
\end{remark}

\subsection{Users with no restriction on their channels}
Now we investigate a SAS-MAC with no restriction on the channels of the users. 
The key ingredient in our analysis is a novel way to bound the probability of error reminiscent of Gallager's error exponent. We show an achievability scheme that allows a positive lower bound on the rates and on $\nu$. This proves that reliable transmission with  an exponential number of users in an exponential asynchronous exponent  is possible. 
We use  an ML decoder sequentially in each block to identify the active user and its message. 

In our results, we use the following notation. The {\it Chernoff distance} between two distributions is defined as
\begin{align}
C(P,Q) := \sup_{0\leq t\leq 1} -\log \left(  \sum_{x}P(x)^{t}Q(x)^{1-t} \right).
\end{align}
We extend this definition and introduce the quantity
\begin{align}
C(P_{i},Q_i,P_{j},Q_j):=\sup_{0\leq t\leq 1} \mu_{i,j}(t),
\end{align}
where 
\begin{align}
\mu_{i,j}(t) := -\log\sum\limits_{x_i,x_j,y}P_{i}(x_i)P_{j}(x_j){Q_i(y|x_i)^{1-t}Q_j(y|x_j)^t}
\label{mu}
\end{align}
is a concave function of $t$. We also define
\begin{align}
C(.,Q_\star,P_{j},Q_j):=\sup_{0\leq t\leq 1}
-\log \left( \sum_{x_j,y}P_{j}(x_j)Q_\star(y)^{1-t}Q_j(y|x_j)^t \right)\notag
\end{align} 
to address the special case with $i=0$ and hence all users are idle.

\begin{theorem}  \label{thm:ach general}
For a SAS-MAC with asynchronous exponent $\alpha$, occupancy exponent $\nu$ and rate $R$, the following region is achievable 
\begin{align}
\bigcup_{n\geq 1}\bigcup_{\substack{P_i\in \mathcal{P}_{\mathcal{X}}}}\bigcap_{i\in [K_n]}\left\{
\begin{matrix}
\nu&<\frac{\alpha}{2}\\
    \nu+R&< B(P_{i},Q_i),  
\\ 2\nu+R&< \inf_{j\neq i}C(P_{j},Q_j,P_{i},Q_i),
\\
\alpha+\nu+R&< C(\,.\,,Q_\star,P_{i},Q_i), 
\end{matrix}\right\}
\end{align}
\end{theorem}
\begin{IEEEproof}
\paragraph*{Codebook generation} Each user $i\in[K_n]$ generates an i.i.d. random codebook according to the distribution $P_{i}$.

\paragraph* {Probability of error analysis} 
The receiver uses the following block by block decoder:
for each block $s\in [A_n]$, the decoder outputs
\begin{align*}
i^* \in \arg \max_{i\in[0:K_n],m\in [M_n]} Q_i(y_s^n|x_i^n(m))\,,
\end{align*}
where  $x_0^n=\emptyset$.

We now find an upper bound on  probability of error given the hypothesis $H^{(1)}$ in~\eqref{eq:def Hyp ThmSameCh} for this decoder as follows
\begin{align*}
&P_e^{(n)}\leq 
 \sum_{i\in[K_n]}\sum_{m\in[2:M_n]}\mathbb{P}\left[\log\frac{Q_i(Y_i^n\mid x_i^n(m))}{Q_i(Y_i^n\mid x_i^n(1))}>0\vert H^{(1)}\right]
\\&+ \sum_{i\in[K_n]} \sum_{\substack{j\in[0:K_n]\\j\neq i}} \sum_{m\in[M_n]}
\mathbb{P}\left[\log \frac{Q_j(Y_i^n|x_j^n(m))}{Q_i(Y_i^n|x_i^n(1))}>0\vert H^{(1)}\right]
\\&+\sum_{s\in[K_n+1:A_n]} \sum_{j\in[K_n]} \sum_{m\in [M_n]}
\mathbb{P}\left[\log\frac{Q_j(Y_s^n|x_j^n(m))}{Q_\star(Y_s^n)}>0\vert H^{(1)} \right]
\\&\leq \sum_{i\in[K_n]}e^{nR}e^{-n\sup_t - \log \mathbb{E}\left[\left( \frac{Q_i(Y_i\mid \overline{X}_i)}{Q_i(Y_i\mid X_i)} \right)^t\right]}
\\&+ \sum_{i\in[K_n]} \sum_{\substack{j\in[0:K_n]\\j\neq i}} e^{nR} e^{{-}n\sup_t {-}\log \mathbb{E}\left[\left(\frac{Q_j(Y_i|X_j)}{Q_i(Y_i|X_i)}\right)^t\right]}
\\&+e^{n\alpha} \sum_{j\in[K_n]} e^{nR} e^{-n\sup_t  {-}\log \mathbb{E}\left[\left(\frac{Q_j(Y_s|X_j)}{Q_{\star}(Y_s)}\right)^t\right]},
\end{align*}
where $P_{X,\overline{X}}(x,x^\prime)=P_{X}(x)P_{X}(x^\prime)$. The last inequality is due to the Chernoff bound.
In order for  each term in the probability of error upper bound to vanish as $n$ grows to infinity, we find the conditions stated in the theorem.  
\end{IEEEproof}

\begin{remark}
Note that  (see Appendix \ref{app:eq:in app 4}):
\begin{subequations}
\begin{align}
  B(P,Q)&:=C(P,Q,P,Q)
=-\log \sum_{x,x^\prime,y}P(x)P(x^\prime)\sqrt{Q(y|x)Q(y|x^\prime)},
\label{supremum}
\\
&C(\,.\,,Q_{\star},P_{j},Q_j)\leq I(P_{j},Q_j)+D\left([P_{j}Q_j] \parallel Q_{\star}\right),
\label{eq:in app 4}
\\
&C(P_{i},Q_{i},P_{j},Q_j)\leq I(P_{j},Q_j)+D\left(P_{i}[P_{j}Q_j] \parallel P_{i}Q_{i}\right),\label{Distance}
\end{align}
\label{eq:all bounds on C and B}
\end{subequations}
where,
due to symmetry, in $C(P,Q,P,Q)$ the  supremum is achieved at the midpoint $t=\frac{1}{2}$, and 
hence  $B(P,Q) = C(P,Q,P,Q)=\mu(\frac{1}{2})$.
The bounds in~\eqref{eq:all bounds on C and B} show that the achievable rates in Theorem~\ref{thm:ach general} are less than the corresponding point-to-point bounds.
\end{remark}
\begin{example}
Consider the SAS-MAC with asynchronous exponent $\alpha$, occupancy exponent $\nu$, and rate $R$ with input-output relationship $Y=\sum_{i\in [K_n]}X_i \oplus Z$ with $Z\sim Ber(\delta)$ for some $\delta\in(0,1/2)$. In our notation 
\begin{align}
Q(y|x)
&=\mathbb{P}[X_i\oplus Z =y|X_i=x]=\mathbb{P}[Z=x\oplus y]
\\&
= \begin{cases}
1-\delta & x \oplus y=0 \ (i.e., \ x=y)\\
  \delta & x \oplus y=1 \ (i.e., \ x\not=y)\\
\end{cases}.
\end{align}
Assume that the input distribution used is $P=Ber(p)$ for some $p\in(0,1/2)$.
The achievability region of this example, based on Theorem~\ref{thm:same channel achievability}, 
 includes the following region
\begin{align}
\bigcup_{\substack{p \in [0,\frac{1}{2}]\\\lambda \in [0,1]}}\left\{
\begin{matrix}
\nu &<\alpha/2 \\
\nu&<p \cdot d(\epsilon_\lambda \parallel \delta)\\
\alpha+R+\nu&<p\cdot d(\epsilon_\lambda \parallel 1-\delta)\\
R+\nu&<h(p*\delta)-h(p)
\end{matrix}\right\},
\label{ex:Th1BSC}
\end{align}
where
\begin{align*}
d(p\parallel q ) &:= p\log\frac{p}{q}+(1-p)\log \frac{1-p}{1-q},
\\
\epsilon_\lambda&:=\frac{\delta^\lambda (1-\delta)^{(1-\lambda)}}{\delta^\lambda (1-\delta)^{(1-\lambda)}+(1-\delta)^\lambda \delta^{(1-\delta)}},
\\
p*q &:=p(1-q)+(1-p)q.
\end{align*} 
Moreover, by assuming $P_{i}=Ber(p_i)$ for all $i\in[K_n]$, we can show that the optimal $t$ in $C(P_{i},Q_i,P_{j},Q_j)=\sup_t \mu_{i,j}(t)$ is equal to $t=1/2$ and hence the achievability region for this channel based on Theorem~\ref{thm:ach general} is given by 
\begin{align*}
\bigcup_{n\geq 1}\bigcup_{\substack{P_i\in \mathcal{P}_{\mathcal{X}}}}\bigcap_{i\in [K_n]}\left\{
\begin{matrix}
\nu&<\frac{\alpha}{2}\\
    \nu+R&<  B(P_{i},Q)=  g(p_i*p_i, \delta),  
\\ 2\nu+R&< \inf_{j\neq i}C(P_{j},Q,P_{i},Q)=\inf_{i\neq j}g(p_i*p_j, \delta),
\\
\alpha+\nu+R&<  C(\,.\,,Q_\star,P_{i},Q) =  g(p_i,  \delta),, 
\end{matrix}\right\}
\end{align*}
 where  
\begin{align*}
g(a,b) &:=-\log\Big(1-a+2a\sqrt{b(1-b)}\Big).
\end{align*}
Finally, by symmetry, we can see that the optimal $p_i=\frac{1}{2}, \forall i\in[K_n]$ and hence  
 $g(\frac{1}{2}, \delta)= -\log\left(1/2+\sqrt{ \delta(1- \delta)}\right) >0$. So on the BSC($ \delta$) strictly positive rates and $\nu$ are achievable. 
In this regrard, the region in Theorem~\ref{thm:ach general} reduces to
\begin{align}
\alpha+\nu+R < -\log\left(1/2+\sqrt{ \delta(1- \delta)}\right). \label{ex:Th2BSC}
\end{align}

The achievable region in~\eqref{ex:Th1BSC}  for $(\alpha, \nu, R)$ is shown in Fig.~\ref{fig:region bernoulli 1}. In addition, the achievable region for $(\alpha, \nu, R)$ with the achievable scheme in~\cite{shahi2016isit} is also plotted in Fig.~\ref{fig:region bernoulli 2} for comparison.  Fig.~\ref{fig:region bernoulli 1nad2} shows that the  achievable scheme in Theorem~\ref{thm:same channel achievability} indeed results in a larger achievable region than the one in Theorem~\ref{thm:ach general}. 

 \begin{figure}
 \centering
 \subfigure[Achievable region in~\eqref{ex:Th1BSC}.]{\includegraphics[width=0.45\textwidth]{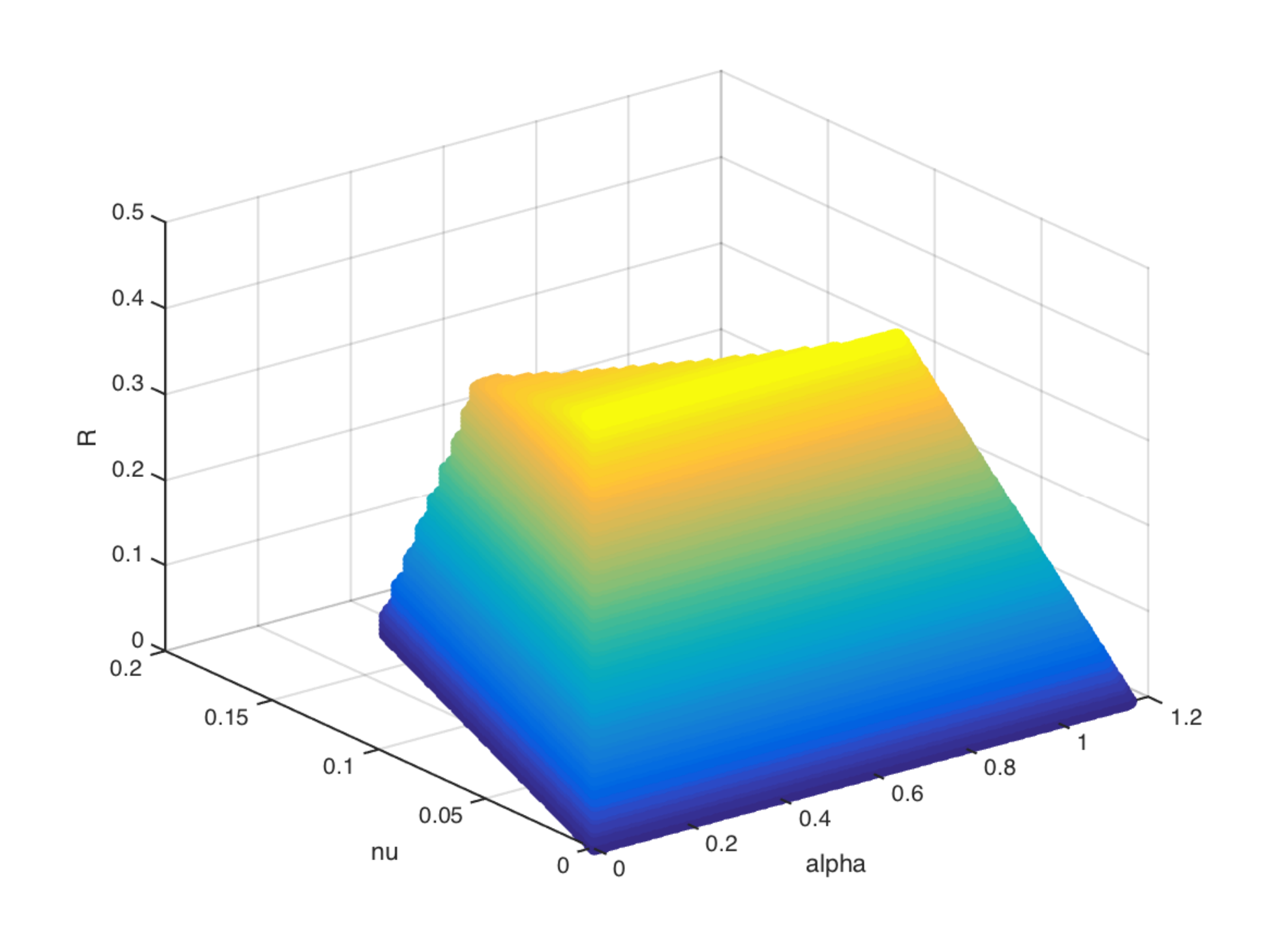}\label{fig:region bernoulli 1}}
 \subfigure[Achievable region in~\eqref{ex:Th2BSC}.]{\includegraphics[width=0.45\textwidth]{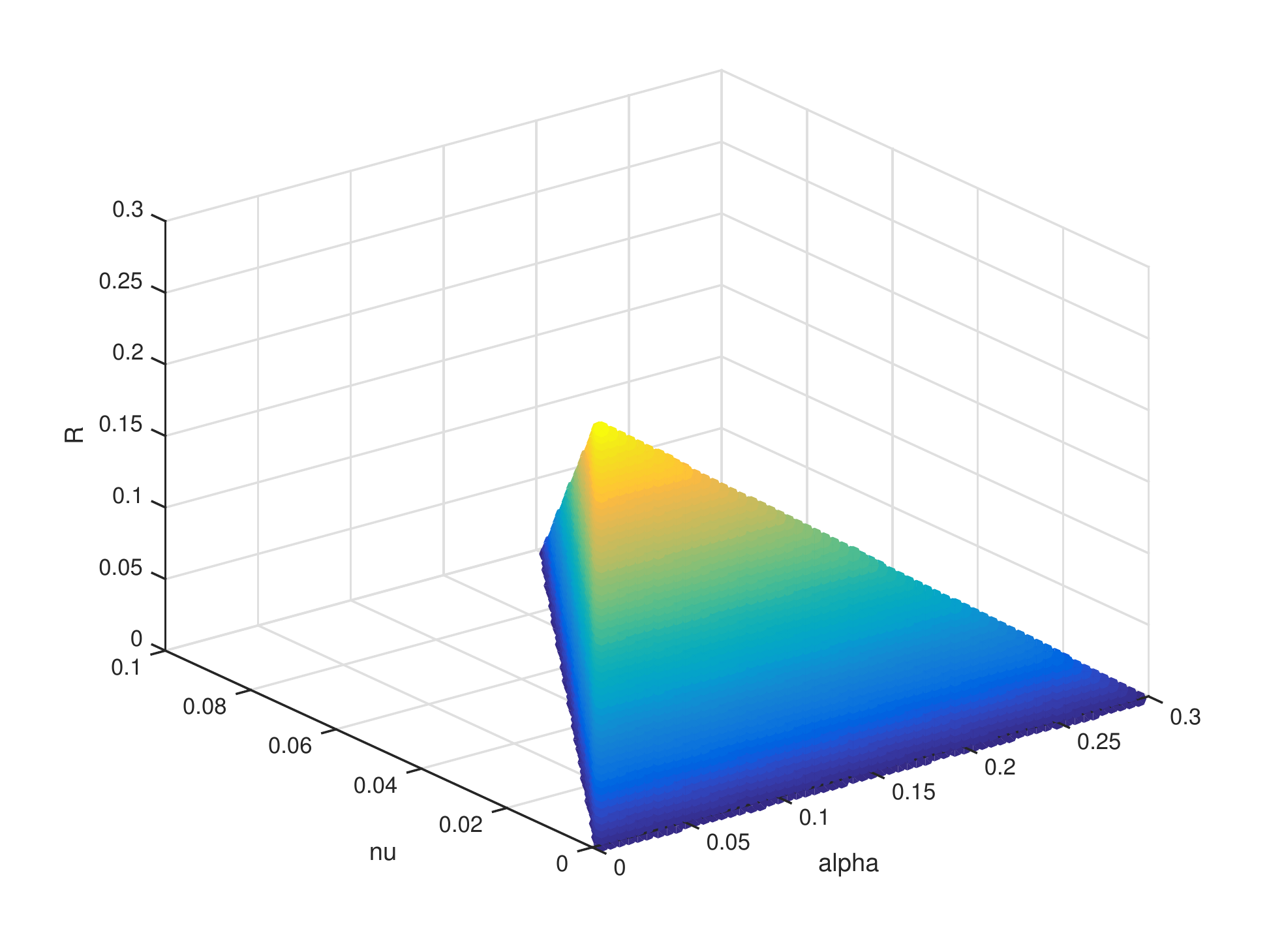}\label{fig:region bernoulli 2}}
\caption{Comparison of the achievable region in Theorem~\ref{thm:same channel achievability} and Theorem~\eqref{thm:ach general}, for the Binary Symmetric Channel with cross over probability $\delta = .11$.}
 \end{figure}
 \label{fig:region bernoulli 1nad2}
\end{example}
 Note that the fact that the achievability region for Theorem~\ref{thm:same channel achievability} is larger than the achievability region of Theorem~\ref{thm:ach general} for identical channels is not surprising. In Theorem~\ref{thm:same channel achievability} we separated the synchronization and decoding steps, whereas in Theorem~\ref{thm:ach general} synchronization and codeword decoding was done the same time, sequentially for each block. The sequential block decoding step result in smaller achievability region in Theorem~\ref{thm:ach general}. 
\subsection{Converse on the capacity region of SAS-MAC}
\label{subsec:Converse}
Thus far, we have provided achievable regions for the SAS-MAC for the cases that different users have identical channels; the case that their channels belong to a set of  size that grows polynomially in the blocklength,  and the case without any restriction on the users' channels. Theorem~\ref{thm:converse sas-mac} next provides a converse to the capacity region of SAS-MAC that holds for any choice of the users' channels.
\begin{theorem}
\label{thm:converse sas-mac}
For the SAS-MAC with asynchronous exponent $\alpha$, occupancy exponent $\nu$ and rate $R$, such that $\nu<\alpha/2$, the following region is impermissible
\begin{align}
\bigcup_{n\geq 1}\bigcup_{\substack{i\in[ K_n]\\P_i\in \mathcal{P}_{\mathcal{X}}\\\lambda_i \in [0,1]}}\left\{\left\{
\begin{matrix}
\nu&>\frac{1}{K_n}\sum_{i=1}^{K_n} D({Q_i}_{\lambda_i}\parallel Q_i \vert P_i),\\
\alpha&>\frac{1}{K_n}\sum_{i=1}^{K_n} D({Q_i}_{\lambda_i}\parallel Q_\star \vert P_i)-(1-\bar{r}_n)(\nu+R)\\
\end{matrix}\right\} \bigcup \left\{R>I(P_i,Q_i)\right\} \right\},\label{eq:thm conv general}
\end{align}
where $\bar{r}_n$ is the infimum probability of error, over all estimators $T$, in distinguishing different  hypothesis ${Q_i}_{\lambda_i}(y^n\vert x_i^n(m)), i\in [K_n], m\in [M_n]$, i.e., 
\begin{align}
\bar{r}_n := \inf_T \frac{1}{K_n M_n} \sum_{i=1}^{K_n}\sum_{m=1}^{M_n} {Q_i}_{\lambda_i}(T\neq i, m \vert x_i^n(m))).
\end{align}
\end{theorem}

\begin{proof}
We first define the following special shorthand notations that we will use throughout this proof
\begin{align}
F_n       &:=M_nK_n, &
\qim(y^n) &:=Q_i(y^n\vert x_i^n(m)),
\\
\py(y^n)  &:=\frac{1}{F_n}\sum_{i=1}^{K_n}\sum_{m=1}^{M_n}\qim(y^n), &
\qiml(y^n)&:= \frac{\left(\qim(y^n)\right)^{\lambda_i} \left(Q_\star^n(y^n)\right)^{1-\lambda_i}}{\sum_{y^n}\left(\qim(y^n)\right)^{\lambda_i} \left(Q_\star^n(y^n)\right)^{1-\lambda_i}},
\\
\pyl(y^n) &:=\frac{1}{F_n}\sum_{i=1}^{K_n}\sum_{m=1}^{M_n}{\qim}_{\lambda_i}(y^n),&
\left(\py\right)_t(y^n)&:=\frac{\left(\py(y^n)\right)^t\left(Q_\star^n(y^n)\right)^{1-t}}{\sum_{y^n}\left(\py(y^n)\right)^t\left(Q_\star^n(y^n)\right)^{1-t}},
\\
Q_\star^n(y^n)&:=\prod_{i=1}^n Q_\star(y_i).
\end{align}
We use the optimal Maximum Likelihood (ML) decoder to find the location of the `active' blocks. In this stage, we are not concerned about the identity or the message of the users. In this regard, the decoder  output is determined via
\begin{align}
\arg\max_{\substack{(l_1,\ldots, l_{K_n})\\l_i\neq l_j, \forall i\neq j\\ l_i\in [A_n], i\in [K_n]}}\sum_{i=1}^{K_n}\log\frac{\py(Y^n_{l_i})}{Q_\star^n(Y^n_{l_i})}.
\end{align}
Given the hypothesis that the users are active in blocks $[K_n]$, denoted by $H^{(2)}$ in~\eqref{eq:def Hyp ThmDifferentCh}, the corresponding error events to the ML decoder are given by
\begin{align}
\left\{\text{error}   \vert H^{(2)} 
\right\}&=\bigcup_{\substack{(l_1,\ldots, l_{K_n})\\\neq (1,\ldots, K_n)}}\left\{\sum_{i=1}^{K_n}\log\frac{\py(Y^n_{l_i})}{Q_\star^n(Y^n_{l_i})}>\sum_{i=1}^{K_n}\log\frac{\py(Y^n_{i})}{Q_\star^n(Y^n_{i})}\right\}\notag\\
&\supseteq \bigcup_{\substack{i\in[K_n]\\j\in[K_n+1:A_n]}}\left\{\log\frac{\py(Y^n_{j})}{Q_\star^n(Y^n_{j})}\geq \log\frac{\py(Y^n_{i})}{Q_\star^n(Y^n_{i})} \right\}\notag\\
&\supseteq \left\{\bigcup_{\substack{j\in[K_n+1:A_n]}}\log\frac{\py(Y^n_{j})}{Q_\star^n(Y^n_{j})}\geq T \right\} \bigcap \left\{ \bigcup_{i\in[K_n]} T\geq \log\frac{\py(Y^n_{i})}{Q_\star^n(Y^n_{i})} \right\},
\label{eq:converse error}
\end{align}
for any $T\in \mathbb{R}$. We restrict ourselves to a subset of $T$'s  and we take $T$ to be
\begin{align}
T:=\frac{1}{F_n}\sum_{i=1}^{K_n}\sum_{m=1}^{M_n}
\left( D\left(\qiml \parallel Q_\star^n\right)-D\left(\qiml \parallel \py \right) \right),
\end{align}
for different $\lambda_i\in [0,1], i\in[K_n]$.

We also find the following lower bounds, which are proved in  Appendix~\ref{appendix:lowerbound converse},
\begin{align}
Q_\star^n&\left[ \log \frac{\py}{Q_\star^n}(Y^n)\geq T\right]\geq 
e^{-\frac{n}{K_n}\left(\sum_{i=1}^{K_n} D({Q_i}_{\lambda_i}\parallel Q_\star \vert P_i)-(R+\nu)(1-\widebar{r}_n)+\frac{h(\widebar{r}_n)}{n}\right)}
,\label{eq:noise prob}
\\
\py&\left[\log \frac{\py}{Q_\star^n}(Y^n)\leq T \right]\geq  
e^{-\frac{n}{K_n}\sum_{i=1}^{K_n} D({Q_i}_{\lambda_i}\parallel Q_i \vert P_i)}
.\label{eq:code prob}
\end{align}
By using the inequalities in~\eqref{eq:noise prob} and~\eqref{eq:code prob}, we find a lower bound on the probability of~\eqref{eq:converse error} as follows:
\begin{align}
&\PP\left[\bigcup_{\substack{j\in[K_n+1:A_n]}}\log\frac{\py(Y^n_{j})}{Q_\star^n(Y^n_{j})}\geq T  \cap  \bigcup_{i\in[K_n]} T\geq \log\frac{\py(Y^n_{i})}{Q_\star^n(Y^n_{i})}  \vert H^{(2)} \right]
\\
&=\PP\left[\bigcup_{\substack{j\in[K_n+1:A_n]}}\log\frac{\py(Y^n_{j})}{Q_\star^n(Y^n_{j})}\geq T \vert H^{(2)}\right]
  \PP\left[\bigcup_{i\in[K_n]} T\geq \log\frac{\py(Y^n_{i})}{Q_\star^n(Y^n_{i})} \vert H^{(2)} \right]
\label{eq:independence blocks}
\\
&=:\PP\left[ Z_1\geq 1\right] \PP[Z_2\geq 1]
\\
&\geq \left( 1-\frac{\var[Z_1]}{\E^2[Z_1]}\right)\left( 1-\frac{\var[Z_2]}{\E^2[Z_2]}\right)
\label{eq:chebyshev conv}
\\
&\geq \left( 1-\frac{1}{\sum_{j=K_n+1}^{A_n}\PP[\xi_j=1]}\right)\left( 1-\frac{1}{\sum_{i=1}^{K_n}\PP[\zeta_i=1]}\right)
\\
&\geq\left(1-e^{-n\alpha +\frac{n}{K_n}\left(\sum_{i=1}^{K_n} D({Q_i}_{\lambda_i}\parallel Q_\star \vert P_i)-(R+\nu)(1-\widebar{r}_n)+\frac{h(\widebar{r}_n)}{n}\right)} \right) \left(1-e^{-n\nu+\frac{n}{K_n}\sum_{i=1}^{K_n} D({Q_i}_{\lambda_i}\parallel Q_i \vert P_i)} \right),\label{eq:conv final lbb}
\end{align}
where~\eqref{eq:independence blocks} follows by independence of $Y_i^n$ and $Y_j^n$ whenever $i\neq j, \forall i, j\in [A_n]$ and the inequality in~\eqref{eq:chebyshev conv} is by Chebyshev's inequality, where we have defined
\begin{align}
Z_1 &:= \sum_{j=K_n+1}^{A_n} \xi_j, &&\xi_j:=Ber\left(Q_\star^n\left[ \log \frac{\py}{Q_\star^n}(Y_j^n)\geq T\right] \right),\\
Z_2 &:= \sum_{i=1}^{K_n} \zeta_i,   &&\zeta_i:=Ber\left(\py\left[\log \frac{\py}{Q_\star^n}(Y_i^n)\leq T \right] \right),
\end{align}
where $\{\xi_j,\zeta_i\}$ are mutually independent.
We can see from~\eqref{eq:conv final lbb} that if
\begin{align}
\nu&>\lim_{n\to \infty}\left\{\frac{1}{K_n}\sum_{i=1}^{K_n} D({Q_i}_{\lambda_i}\parallel Q_i \vert P_i)\right\},\\
\alpha&>\lim_{n\to \infty}\left\{\frac{1}{K_n}\sum_{i=1}^{K_n} D({Q_i}_{\lambda_i}\parallel Q_\star \vert P_i)-(1-\bar{r}_n)(\nu+R)\right\},
\end{align}
then the probability of error is strictly bounded away from zero and hence it is impermissible.
Moreover, the usual converse bound on the rate of a synchronous channel also applies to any asynchronous channel  and hence the region where $R>I(P,Q)$ is also impermissible. 
This concludes the proof.
\end{proof}
{ It should be noted that even though the expression~\eqref{eq:thm conv general} involves a union over all blocklengths $n$, in order to compute this bound, we only need to optimize with respect to $P_i, i\in [K_n]$ (as opposed to $P^n$ in the conventional $n$-letter capacity expressions). However, since we still have exponential (in blocklength $n$)  number of users, and in theory we have to optimize all of their distributions, we need to take the union with respect to all blocklengths.}
\subsection{Converse on the number of users in a SAS-MAC}
{ In previous sections and in our achievability schemes, we restricted ourselves to the region where $\nu<\frac{\alpha}{2}$ to be able to simplify the analysis. However, an interesting question is that  irrespective of the achievability scheme and the decoder, how large a $\nu$ we can have in the network. Theorem~\ref{thm:conv nu} provides a converse bound on the value of $\nu$ such that for $\nu>\alpha$, not even reliable synchronization is not possible. 

\begin{theorem} ~\label{thm:conv nu}
For a SAS-MAC with asynchornous exponent $\alpha$ and occupancy exponent $\nu : \nu>\alpha$, 
reliable
synchronization is not possible,  i.e., even with $M=1$, one has $P_e^{(n)}>0.$ 
\end{theorem}

\begin{figure}[htbp]
\centering
\includegraphics[width=.8\textwidth]{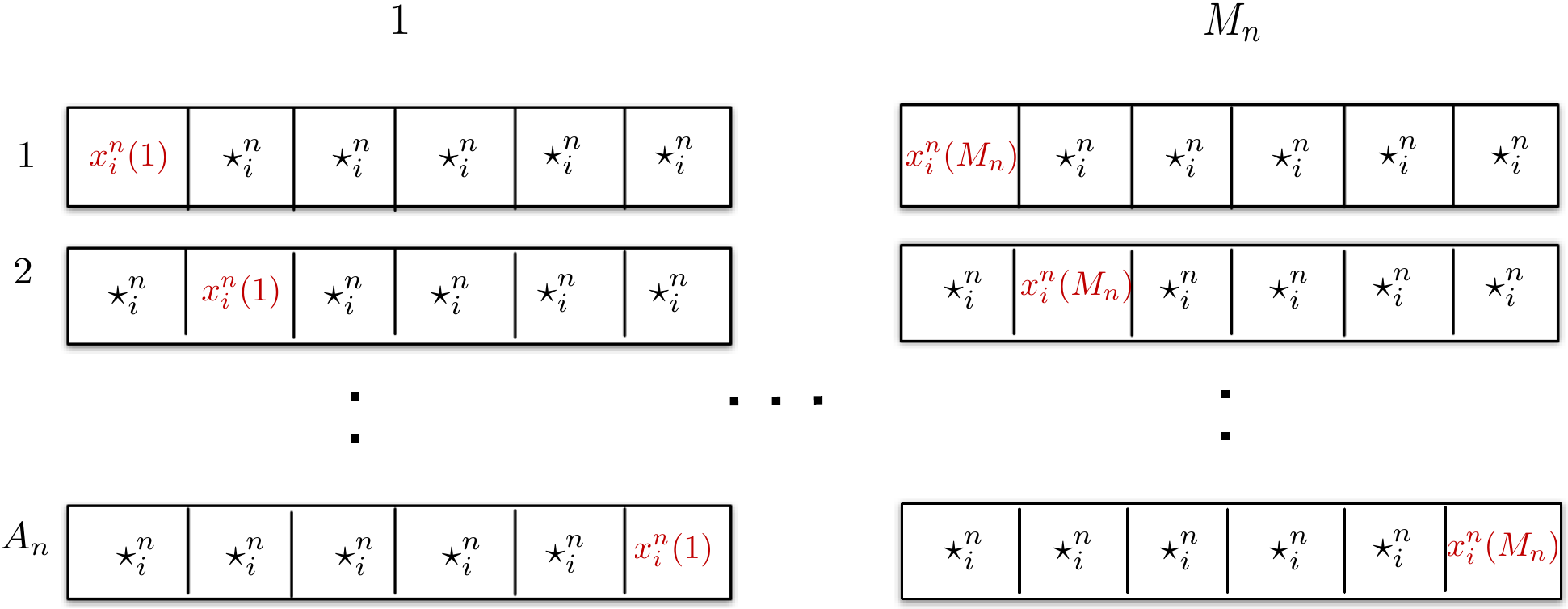}
\caption{Extended codebook.}
\label{extended_codebook}
\end{figure}

\begin{IEEEproof}
User $i\in [{K_n}]$ has a codebook with $M_n=e^{nR}$ codewords of length $n$. 
Define for $i\in [{K_n}]$ an `extended codebook' consisting of $A_nM_n$ codewords of length $nA_n$ constructed such that $\forall m\in [M_n]$ and $\forall t_i \in [A_n]$
\begin{align*}
\widetilde{X}_i^{nA_n}(m_i,t_i):=\big[\star_i^{n(t_i-1)} \,f_i(m_i)\, \star_i^{n(A_n-t_i)}\big],
\end{align*} 
as depicted in Fig.~\ref{extended_codebook}. 
By using Fanno's inequality, i.e., $H({X}_1^{nA_n},\ldots,{X}_{K_n}^{nA_n}|Y^{nA_n}) \leq n \epsilon_n : \epsilon_n\to 0$ as $n\to\infty$, for any codebook of length $nA_n$ we have
\begin{align*}
&H({X}_1^{nA_n},\ldots, {X}_{K_n}^{nA_n})
=H(m_1,t_1,\ldots,m_{K_n},t_{K_n})
\\&=n K_n\left( \alpha+R\right)
\\&=H({X}_1^{nA_n},\ldots,{X}_{K_n}^{nA_n}|Y^{nA_n})+I({X}_1^{nA_n},\ldots,{X}_{K_n}^{nA_n};Y^{nA_n})
\\&\leq n\epsilon_n+ne^{n\alpha} \ |\mathcal{Y}|
\quad \Longleftrightarrow
\\&\nu+\frac{\log\left(1+\frac{1}{\alpha K_n}\sum_{i\in[K_n]}R_i \right)}{n}\leq \alpha+
\frac{\log\left(1+\frac{\epsilon_n}{e^{n\alpha}\mid\mathcal{Y}\mid}\right)}{n},
\end{align*}
where $\frac{\log\left(1+\frac{1}{\alpha K_n}\sum_{i\in[K_n]}R_i \right)}{n} \geq 0$ and $\frac{\log (1+\frac{\epsilon_n}{e^{n\alpha}\mid\mathcal{Y}\mid})}{n} \geq 0$ vanish as $n$ goes to infinity. This implies that $\nu \leq \alpha$ is a necessary condition for reliable communications.
In other words, for $\nu > \alpha$ not even synchronization (i.e., $M_n=1, \forall i\in[K_n]$)  is possible.
\end{IEEEproof}
}

\section{Discussion and Conclusion}
\label{sec:conclusion}
In this paper we studied a Strongly Asynchronous and Slotted Massive Access Channel (SAS-MAC) where $K_n:=e^{n\nu}$ different users transmit a randomly selected message among $M_n:=e^{nR}$ ones within a strong asynchronous window of length $A_n:=e^{n\alpha}$ blocks of $n$ channel uses each.  We found  inner and outer bounds on the $(R, \alpha, \nu)$ tuples. Our analysis is based on a global probability of error in which we required all users messages and identities to be jointly correctly decoded. Our results are focused on the region $\nu<\frac{\alpha}{2}$, where the probability of user collisions in vanishing.  We  proved in Theorem~\ref{thm:conv nu}  that for the region $\nu>\alpha$, not even synchronization is possible. Hence, we would like to take this chance to discuss some of the difficulties that one may face in analyzing the region $ \frac{\alpha}{2}\leq \nu\leq \alpha$. 

As we have mentioned before, for the region $\nu<\frac{\alpha}{2}$, with probability  
$\frac{\binom{A_n}{K_n}}{(A_n)^{K_n}}$ which approaches to one as blocklength $n$ goes to infinity, the users transmit in distinct blocks. Hence, in analyzing the probability of error of our achievable schemes, we could safely condition on the hypothesis that users are not colliding.  For the region $\frac{\alpha}{2}\leq \nu\leq \alpha$, we lose this simplifying assumption. In particular, based on Lemma~\ref{lemma:no overlapping} (proved in the Appendix~\ref{appendix: no overlapping}), for the region $\frac{\alpha}{2}\leq \nu\leq \alpha$, the probability of every arrangement of users is itself vanishing in the blocklength. 
\begin{lemma}
\label{lemma:no overlapping}
For the region $\frac{\alpha}{2}\leq \nu\leq \alpha$ the non-colliding arrangement of users has the highest probability among all possible arrangements, yet, the probability of this event is also vanishing as blocklength $n$ goes to infinity. 
\end{lemma}
As a consequence of Lemma~\ref{lemma:no overlapping}, one needs to propose an achievable scheme that accounts for every possible arrangement and collision of users and drives the probability of error in all (or most) of the hypothesis to zero. It is also worth noting that the number of possible hypotheses is doubly exponential in the blocklength. Finally, it is worth emphasizing the reason why the authors in~\cite{chandar2015note} can get to $\nu\leq \alpha$. In~\cite{chandar2015note} the authors require the recovery of the messages of a {\it large fraction} of users and they also require the per-user probability of error to be vanishing. To prove whether or not strictly positive $(R, \alpha, \nu)$ are possible in the region $\frac{\alpha}{2}\leq \nu\leq \alpha$, with vanishing global probability of error, is an open problem.

\appendix
\subsection{Proof of Lemma~\ref{lemma:graph}} ~\label{app:lemma graph proof}
 We first consider the case that r is an even number and then prove
 \begin{align}
r\!\left(n_k\right)^{\frac{r}{2}-1}\left( G(c_1)+\ldots G(c_{N_{r,k}})\right) \!\leq\! \frac{N_{r,k}r}{n_k}\! \left(a_1^2+\ldots+ {a_{n_k}}^2\right)^{\frac{r}{2}}\!.\label{eq:main ineq}
\end{align}
We may drop the subscripts and use $N:=N_{r,k}$ and $n:=n_k$ in the following for notational ease.
Our goal is to expand the right hand side (RHS) of~\eqref{eq:main ineq} such that all elements have coefficient $1$. Then, we parse these elements into $N$ different groups (details will be provided later) such that using the AM-GM inequality (i.e., $ n\left(\prod_{i=1}^n a_i\right)^{\frac{1}{n}}\leq \sum_{i=1}^n a_i$) on each group, we get one of the $N$ terms on the LHS of~\eqref{eq:main ineq}. Before stating the rigorous proof, we provide an example of this strategy for the graph with $k=4$ vertices shown in Fig.~\ref{fig:G6}. In this example, we consider the Lemma for $r=4$ cycles (for which $N=3$). 
\begin{figure}[htbp]
\centering
\includegraphics[width=.15\textwidth]{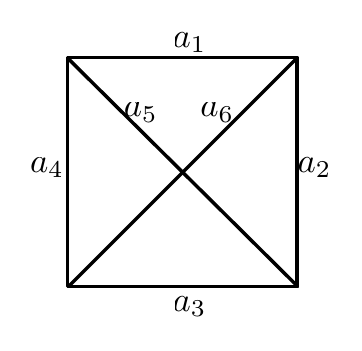}
\caption{A complete graph with $4$ vertices}
\label{fig:G6}
\end{figure}

We may expand the RHS in~\eqref{eq:main ineq} as
\begin{align*}
&\qquad2\left(a_1^2+\ldots +a_6^2 \right)^2
=\Theta_1+\Theta_2+\Theta_3,\\
&\Theta_1\!=\!\big\{a_1^4+a_2^4+a_3^4+a_4^4+a_1^2a_3^2+a_1^2a_3^2+a_2^2a_4^2+a_2^2a_4^2\\
&+a_1^2a_2^2+a_1^2a_2^2+a_1^2a_2^2+a_1^2a_2^2+a_1^2a_4^2+a_1^2a_4^2+a_1^2a_4^2+a_1^2a_4^2
\\&+a_2^2a_3^2+a_2^2a_3^2+a_2^2a_3^2+a_2^2a_3^2+a_3^2a_4^2+a_3^2a_4^2+a_3^2a_4^2+a_3^2a_4^2\big\}
\\&\Theta_2\!=\!\big\{a_1^4+a_6^4+a_3^4+a_5^4+a_5^2a_6^2+a_5^2a_6^2+a_1^2a_3^2+a_1^2a_3^2
\\&
+a_1^2a_6^2+a_1^2a_6^2+a_1^2a_6^2+a_1^2a_6^2+a_1^2a_5^2+a_1^2a_5^2+a_1^2a_5^2+a_1^2a_5^2
\\&+a_3^2a_6^2+a_3^2a_6^2+a_3^2a_6^2+a_3^2a_6^2+a_3^2a_5^2+a_3^2a_5^2+a_3^2a_5^2+a_3^2a_5^2\big\}
\\&\Theta_3\!=\!\big\{a_4^4+a_5^4+a_2^4+a_6^4+a_5^2a_6^2+a_5^2a_6^2+a_2^2a_4^2+a_2^2a_4^2
\\&
+a_4^2a_5^2+a_4^2a_5^2+a_4^2a_5^2+a_4^2a_5^2
+a_4^2a_6^2+a_4^2a_6^2+a_4^2a_6^2+a_4^2a_6^2
\\&
+a_2^2a_5^2+a_2^2a_5^2+a_2^2a_5^2+a_2^2a_5^2+a_2^2a_6^2+a_2^2a_6^2+a_2^2a_6^2+a_2^2a_6^2\big\}.
\end{align*}
It can be easily seen that if we use the AM-GM inequality on $\Theta_1$, $\Theta_2$ and $\Theta_3$, we can get the lower bound equal to $24(a_1a_2a_3a_4), 24(a_1a_6a_3a_5)$ and $24(a_4a_5a_2a_6)$, respectively where $rn^{\frac{r}{2}-1}=24$ and hence~\eqref{eq:main ineq} holds in this example.

We proceed to prove  Lemma~\ref{lemma:graph} for arbitrary $k$ and (even) $r\geq 2$. We propose the following scheme   to group the elements on the RHS of~\eqref{eq:main ineq} and then we prove that this grouping indeed leads to the claimed inequality in the Lemma. 

\paragraph*{ Grouping scheme} For each cycle $c_i =\{a_{i_1}\ldots , a_{i_r}\}$, we need a group of elements, $\Theta_i$, from the RHS of~\eqref{eq:main ineq}. In this regard, we consider all possible subsets of the edges of cycle $c_i$ with $1:\frac{r}{2}$ elements (e.g. $\left\{\{a_{i_1}\},\ldots \{a_{i_1},a_{i_2}\},\ldots \{a_{i_1}\ldots, a_{i_{r/2}}\},\ldots\right\} $). For each one of these subsets, we find the respective elements from the RHS of~\eqref{eq:main ineq} that is the multiplication of the elements in that subset. For example, for the subset $\{a_{i_1},a_{i_2},a_{i_3}\}$, we consider the elements like $a_{i_1}^{n_{i_1}}a_{i_2}^{n_{i_2}}a_{i_3}^{n_{i_3}}$ for all possible $n_{i_1},n_{i_2},n_{i_3}>0$ from the RHS of~\eqref{eq:main ineq}. However, note that we do not assign all such elements to cycle $c_i$ only. If there are $l$ cycles of length $r$ that all contain $\{a_{i_1},a_{i_2},a_{i_3}\}$, we should assign $\frac{1}{l}$ of the elements like $a_{i_1}^{n_{i_1}}a_{i_2}^{n_{i_2}}a_{i_3}^{n_{i_3}}, \ n_{i_1},n_{i_2},n_{i_3}>0$ to cycle $c_i$ (so that we can assign the same amount of elements to other cycles with similar edges). 

We state some facts, which  can be easily verified:

{\bf Fact 1.} In a complete graph $K_k$, there are $N=N_{r,k}=\binom{k}{r}\frac{(r-1)!}{2}$ cycles of length $r$. 

{\bf Fact 2.} By expanding the RHS of~\eqref{eq:main ineq} such that all elements have coefficient $1$, we end up with $\left(\frac{N r}{n}\right) n^{\frac{r}{2}}$ elements.

{\bf Fact 3.} Expanding the RHS of~\eqref{eq:main ineq} such that all elements have coefficient $1$, and finding their product yields
\[\left(a_1\times\ldots\times a_n \right)^{\left(\frac{Nr}{n}\right)rn^{ \frac{r}{2}-1}}.\]

{\bf Fact 4.} In above grouping scheme each element on the RHS of~\eqref{eq:main ineq} is summed in exactly one group. Hence, by symmetry and Fact 2, each group is the sum of $r n^{\frac{r}{2}-1}$ elements.

Now, consider any two cycles $c^{(e)}_i=\{a_{i_1},\ldots , a_{i_r}\},c^{(e)}_j=\{a_{j_1},\ldots , a_{j_r}\} $. Assume that using the above grouping scheme, we get the group of elements $\Theta_i,\Theta_j$ (where by fact 3 each one is the sum of $r n^{\frac{r}{2}-1}$ elements).
If we apply the AM-GM inequality on each one of the two groups, we get 
\begin{align*}
\Theta_i\geq r n^{\frac{r}{2}-1} \left( a_{i_1}^{n_{i_1}}\times \ldots\times a_{i_r}^{n_{1_r}}  \right)^{\left(\frac{1}{r n^{\frac{r}{2}-1}}\right)},  \\
\Theta_j\geq r n^{\frac{r}{2}-1} \left( a_{j_1}^{n_{j_1}}\times \ldots\times a_{j_r}^{n_{j_r}}  \right)^{\left(\frac{1}{r n^{\frac{r}{2}-1}}\right)},
\end{align*}
where $\prod_{t=1}^r a_{i_t}^{n_{i_t}}$ is the product of the elements in $\Theta_i$.
 By symmetry of the grouping scheme for different cycles, it is obvious that $\forall t\in[r], n_{i_t}=n_{j_t}$. Hence $ n_{i_t}=n_{j_t}=p_t,\forall i,j\in [N]$. 
 i.e., we have
 \begin{align}
  \Theta_i &\geq r n^{\frac{r}{2}-1} \left( a_{i_1}^{p_1}\times \ldots\times a_{i_r}^{p_r}  \right)^{\left(\frac{1}{r n^{\frac{r}{2}-1}}\right)}\label{eq:p}.
 \end{align}

 By symmetry of the grouping scheme over the elements of each cycle, we also get that $n_{i_k}=n_{i_l}=q_i,\forall k,l\in [r]$. i.e.
 \begin{align}
 \Theta_i \geq r n^{\frac{r}{2}-1} \left( a_{i_1}^{q_i}\times \ldots\times a_{i_r}^{q_i}  \right)^{\left(\frac{1}{r n^{\frac{r}{2}-1}}\right)}.\label{eq:q}
 \end{align}
 It can be seen from~\eqref{eq:p} and~\eqref{eq:q} that all the elements of all groups have the same power  $n_{i_t}=p,\forall i\in[N], t\in [r]$. i.e.,
  \begin{align*}
 \Theta_i &\geq r n^{\frac{r}{2}-1} \left( a_{i_1}^{p}\times \ldots\times a_{i_r}^{p}  \right)^{\left(\frac{1}{r n^{\frac{r}{2}-1}}\right)}.
 \end{align*}
 Since each element on the RHS of~\eqref{eq:main ineq} is assigned to one and only one group and since $\prod_{t=1}^r a_{i_t}^{n_{i_t}}= \prod_{t=1}^r a_{i_t}^p$ is the product of the elements of each group $\Theta_i$, the product of all elements in $\Theta_1+\ldots +\Theta_{N}$ (which is equal to product of the elements in the expanded version of the RHS of~\eqref{eq:main ineq}) is
 $\prod_{i=1}^{N}\prod_{t=1}^r a_{i_t}^p$.
 
 In addition, since each $a_i$ appears in exactly $\frac{Nr}{n}$ of the cycles, by Fact 3 and a double counting argument, we have 
 \[p\times \frac{Nr}{n}=\left( \frac{Nr}{n}\right) rn^{ \frac{r}{2}-1},\]
 and hence $p=rn^{ \frac{r}{2}-1}$.
 Hence, the lower bound of the AM-GM inequality on the $\Theta_1+\ldots+ \Theta_{N}$, will result in
  \[rn^{ \frac{r}{2}-1} G(c_1)+\ldots+ rn^{ \frac{r}{2}-1} G(c_{N_r}),\] and  the Lemma is proved for even $r$.
 
 For odd values of $r$, the problem that may arise by using the grouping strategy in its current form, is when $r<\frac{k}{2}$. In this case, some of the terms on the RHS of~\eqref{eq:main ineq} may contain   multiplication of $a_i$'s that are not present in any of the $G(c_i)$'s. To overcome this, take both sides to the power of $2m$ for the smallest $m$ such that $rm>\frac{k}{2}$. Then the RHS of~\eqref{eq:main ineq} is at most the multiplication of $rm$ different $a_i$'s and on the LHS of~\eqref{eq:main ineq},  there are $2m$ cycles of length $r$ multiplied together. By our choice of $2m$, now, all possible combinations of $a_i$'s on the RHS are present in at least one cycle multiplication in the LHS. Hence, it is straightforward to use the same strategy as even values of $r$ to prove the theorem for the odd values of $r$.

\subsection{Proof of~\eqref{eq:ident ach}}~\label{app:ident ach}
By Lemma~\ref{lemma:graph} and~\eqref{eq:achieve graph} we can write
\begin{align}
P_e^{(n)}&\leq \sum_{r=2}^{A_n}\sum_{c\in C_{A_n}^{(r)}}G(c)\notag\\
&\leq \sum_{r=2}^{A_n} \frac{N_{r,A_n}}{\left({n_{A_n}}\right)^{\frac{r}{2}}}\left( a_1^2+\ldots+a_{n_{A_n}}^2\right)^{r/2}\notag\\
&\leq \sum_{r=2}^{A_n} 4^r\left( \sum_{\substack{1\leq i<j\leq A_n}}e^{-2nB(P_i,P_j)}\right)^{r/2}\label{eq:nN}\\
&\leq \frac{16\left(\sum_{\substack{1\leq i<j\leq A_n}}e^{-2nB(P_i,P_j)}\right)}{1-4\sqrt{\sum_{\substack{1\leq i<j\leq A_n}}e^{-2nB(P_i,P_j)}}}\label{eq:achieve final ub},
\end{align} 
where~\eqref{eq:nN} is by Fact 1 (see Appendix~\ref{app:lemma graph proof}) and
\[\frac{N_{r,A_n}}{\left( {n_{A_n}}\right)^{r/2}}=\frac{\binom{A_n}{r}(r-1)! /2}{\left(\binom{A_n}{2}\right)^{r/2}}\leq 4^{r}.\]
As the result,~\eqref{eq:achieve final ub} will go to zero as $n$ goes to infinity if
\[\lim_{n\to \infty}  \sum_{\substack{1\leq i<j\leq A_n}}e^{-2nB(P_i,P_j)} = 0. \]
 \subsection{Proof of~\eqref{eq:ident conv}}\label{app:ident conv}
 We upper bound the denominator of~\eqref{eq:conv lb} by  
\begin{align}
&\!\PP[\xi_{i,j}, \xi_{i,k}]=\PP\left[ \log\frac{P_i(X^n_j)}{P_j(X^n_j)}+\log\frac{P_j(X_i^n)}{P_i(X_i^n)}\geq 0
\cap\ \log\frac{P_i(X^n_k)}{P_k(X^n_k)}+\log\frac{P_k(X_i^n)}{P_i(X_i^n)}\geq 0 \right]\notag
\\&
\leq \PP \left[ \log\frac{P_i(X^n_j)}{P_j(X^n_j)}+\log\frac{P_j(X_i^n)}{P_i(X_i^n)}
+\log\frac{P_i(X^n_k)}{P_k(X^n_k)}+\log\frac{P_k(X_i^n)}{P_i(X_i^n)} \geq 0 \right]\notag
\\&
\leq \exp\Bigg\{n \inf_t  
\log\left( \mathbb{E}\left[ \left(\!\frac{P_i(X^n_j)}{P_j(X^n_j)} \cdot \frac{P_j(X_i^n)}{P_i(X_i^n)} \cdot \frac{P_i(X^n_k)}{P_k(X^n_k)}\cdot \frac{P_k(X_i^n)}{P_i(X_i^n)}\right)^t \right]\right) \Bigg\}\notag
\\&
\!\!\leq \!\exp\left\{\!n \log \mathbb{E}\left[\! \left(\frac{P_i(X^n_j)}{P_j(X^n_j)}\!\cdot\!\frac{P_j(X_i^n)}{P_i(X_i^n)}\cdot \! \frac{P_i(X^n_k)}{P_k(X^n_k)} \cdot\! \frac{P_k(X_i^n)}{P_i(X_i^n)}\!\right)^{\frac{1}{2}} \!\right] \right\}\notag
\\&= \exp\left\{-n B({P_i},{P_j})-nB({P_j},P_k)-nB({P_i},{P_k}) \right\}.\label{eq:covariance}
\end{align}
An upper bound for $\PP\left[\xi_{i,j},\xi_{k,l} \right]$ can be derived similarly.


\subsection{Proof of~\eqref{eq:in app 4}}
\label{app:eq:in app 4}

We find an upper bound on $C(\,.\,,Q_\star,P_{i},Q_i)$ 
by noting that $\mu_{0,i}(t)$ defined in~\eqref{mu} is concave in $t$ with $\mu_{0,i}(1)=0$ and 
\begin{align*}
\frac{\partial \mu_{0,i}(t)}{\partial t}\mid_{t=1}
=-I(P_{i},Q_i)-D([P_{i}Q_i] \parallel Q_\star) \leq 0.
\end{align*}
 Hence $\mu_{0,i}(t)$ is always less than $(I(P_{i},Q_i)+D([P_{i}Q_i] \parallel Q_\star))(1-t)$ and that for $0\leq t \leq 1$ it is always less than $I(P_{i},Q_i)+D([P_{i}Q_i] \parallel Q_\star)$. The inequality in~\eqref{Distance}  follows similarly.
 

\subsection{Proof of~\eqref{eq:noise prob} and~\eqref{eq:code prob}}~\label{appendix:lowerbound converse}
Before deriving lower bounds on~\eqref{eq:noise prob} and~\eqref{eq:code prob}, we note that by the Type-counting Lemma~\cite{csiszarbook}, at the expense of a small decrease in the rate  (which vanishes in the limit for large blocklength) we may restrict our attention to constant composition codewords. Henceforth, we assume that the composition of the codewords for user $i\in [K_n]$ is given by $P_i$. Moreover, to make this paper self-contained, we restate  the following Lemmas that we use in the rest of the proof.
\begin{lemma}[Compensation Identity] For arbitrary $\pi_i: \sum_{i=1}^K \pi_i=1$ and arbitrary probability distribution functions $P_i\in \mathcal{P}_{\mathcal{X}}, i\in[K]$, we define $\bar{P}(x)=\sum_{i=1}^K \pi_i P_i(x)$. Then for any probability distribution function $R$ we have:
\begin{align}
&D\left(\bar{P}\parallel R\right)+\sum_{i=1}^K\pi_i D\left(P_i\parallel \bar{P}\right)
=\sum_{i=1}^K\pi_iD(P_i||R).
\end{align}
\end{lemma}
\begin{lemma}[Fano]
\label{lemma:fano}
Let $F$ be an arbitrary set of size $N$. 
For $\bar{P}=\frac{\sum_{\theta\in F} P_\theta}{N}$ we have
\begin{align}
\frac{1}{N}\sum_{\theta\in F}D\left(P_\theta\parallel \bar{P}\right) \geq \left(1-\bar{r}\right)\log\left(N(1-\bar{r})\right) +\bar{r}\log\left( \frac{N\bar{r}}{N-1}\right),
\end{align}
 where
\begin{align}\bar{r}:=\inf_T \frac{1}{N}\sum_{\theta\in F} P_\theta\left\{T\neq \theta\right\}\end{align}
in which the infimum is taken over all possible estimators $T$.
\end{lemma}
We now continue with the proof of~\eqref{eq:noise prob}. 
Using the Chernoff bound we can write
\begin{align}
Q_\star^n&\left[ \log \frac{\py}{Q_\star^n}(Y^n)\geq \frac{1}{F_n}\sum_{i=1}^{K_n}\sum_{m=1}^{M_n}D\left(\qiml \parallel Q_\star^n\right)-D\left(\qiml \parallel \py \right)\right]\overset{\text{Chernoff}}{\doteq } e^{-\sup_t A(t)}.
\end{align}
The Chernoff bound exponent, $\sup_t A(t)$, is expressed and simplified as follows
\begin{align}
A(t)&:= \frac{t}{F_n}\sum_{i=1}^{K_n}\sum_{m=1}^{M_n}D\left(\qiml \parallel Q_\star^n\right)-D\left(\qiml \parallel \py \right)-\log\mathbb{E}_{Q_\star^n}\left[\left( \frac{\py}{Q_\star^n}(Y^n)\right)^t \right]\notag\\
&= \frac{t}{F_n}\sum_{i=1}^{K_n}\sum_{m=1}^{M_n} \sum_{y^n}\qiml(y^n)\log\frac{\py(y^n)}{Q_\star^n(y^n)}-\log\sum_{y^n}\left(\py(y^n)\right)^t\left(Q_\star^n(y^n)\right)^{1-t} \notag\\
&= \frac{1}{F_n}\sum_{i=1}^{K_n}\sum_{m=1}^{M_n} \sum_{y^n}\qiml(y^n)\log\frac{\frac{\left(\py(y^n)\right)^t\left(Q_\star^n(y^n)\right)^{1-t}}{\sum_{y^n}\left(\py(y^n)\right)^t\left(Q_\star^n(y^n)\right)^{1-t}}}{Q_\star^n(y^n)} \notag\\
&= \frac{1}{F_n}\sum_{i=1}^{K_n}\sum_{m=1}^{M_n} \sum_{y^n}\qiml(y^n)\log\frac{\left(\py\right)_t(y^n)}{Q_\star^n(y^n)} \notag\\
&= \frac{1}{F_n}\sum_{i=1}^{K_n}\sum_{m=1}^{M_n} D\left( \qiml\parallel Q_\star^n\right)-D\left( \qiml\parallel \left(\py\right)_t\right)\notag\\
&= \frac{1}{K_n}\sum_{i=1}^{K_n} nD\left( {Q_i}_{\lambda_i}\parallel Q_\star\vert P_i\right)- \frac{1}{F_n}\sum_{i=1}^{K_n}\sum_{m=1}^{M_n} D\left( \qiml\parallel \left(\py\right)_t\right),\label{eq:constant comp result}
\end{align}
where~\eqref{eq:constant comp result} is the result of constant composition structure of the codewords.
As a result,
\begin{align}\sup_t A(t)=\frac{1}{K_n}\sum_{i=1}^{K_n} nD\left( {Q_i}_{\lambda_i}\parallel Q_\star\vert P_i\right)-\inf_t \left\{\frac{1}{F_n}\sum_{i=1}^{F_n}D\left( \qiml\parallel \left(\py\right)_t\right)\right\}.\end{align}
Moreover,
\begin{align}
 &\inf_t\frac{1}{F_n}\sum_{i=1}^{K_n}\sum_{m=1}^{M_n}D\left( \qiml\parallel \left(\py\right)_t\right)\\
&=\inf_t\frac{1}{F_n}\sum_{i=1}^{K_n}\sum_{m=1}^{M_n}\sum_{y^n}\qiml(y^n)\log\frac{\qiml(y^n)}{\left(\py\right)_t(y^n)}\\
&=\inf_t\frac{1}{F_n}\sum_{i=1}^{K_n}\sum_{m=1}^{M_n}\sum_{y^n}\qiml(y^n)\log\frac{\qiml(y^n)}{\left(\py\right)_t(y^n)}\cdot\frac{{\pyl}(y^n)}{\pyl(y^n)}\\
&=\inf_t\frac{1}{F_n}\sum_{i=1}^{K_n}\sum_{m=1}^{M_n} D\left(\qiml \parallel \pyl\right)+D\left(\pyl\parallel \left(\py\right)_t \right)\\
&\geq \frac{1}{F_n}\sum_{i=1}^{K_n}\sum_{m=1}^{M_n} D\left(\qiml \parallel \pyl\right).
\end{align}
Note that $\pyl$ is the average of $\qiml$ over $m, i$'s ($m\in [M_n],i\in[K_n]$) and hence based on Lemma~\ref{lemma:fano}, we have
\begin{align}
\frac{1}{F_n}\sum_{i=1}^{K_n}\sum_{m=1}^{M_n} D\left(\qim \parallel \pyl\right)&\geq  \left(1-\widebar{r}_n\right)\log\left(F_n(1-\widebar{r}_n)\right) +\widebar{r}_n\log\left( \frac{F_n\widebar{r}_n}{F_n-1}\right)
\\&\geq \left(1-\widebar{r}_n\right)\log F_n - h(\bar{r}_n),
\end{align}
where $h(.)$ is the binary entropy function. 
As a result
\begin{align}
\sup_t A(t)& \leq \frac{1}{K_n}\sum_{i=1}^{K_n} nD({Q_i}_{\lambda_i}\parallel Q_\star \vert P_i)-n(R+\nu)(1-\widebar{r}_n)+h(\bar{r}_n).
\end{align}
Now we continue with the proof of~\eqref{eq:code prob}. Again, using the Chernoff bound we have
\begin{align}
&\py\left[\log \frac{P_{Y^n}}{Q_\star^n}(Y^n)\leq\frac{1}{F_n}\sum_{i=1}^{K_n}\sum_{m=1}^{M_n}D\left(\qiml \parallel Q_\star^n\right)-D\left(\qiml \parallel \py \right) \right]\\
&=\py\left[\log \frac{Q_\star^n}{P_{Y^n}}(Y^n)\geq  \frac{1}{F_n}\sum_{i=1}^{K_n}\sum_{m=1}^{M_n}D\left(\qiml \parallel \py \right)-D\left(\qiml \parallel Q_\star^n\right) \right] \doteq e^{-\sup_t B(t)},
\end{align}
where
\begin{align}
\sup_tB(t)&:= \sup_t \frac{t}{F_n}\sum_{i=1}^{K_n}\sum_{m=1}^{M_n} \sum_{y^n}\qiml(y^n)\log \frac{Q_\star^n(y^n) \py(y^n)}{\py(y^n)} -\log \sum_{y^n}\left(Q_\star^n(y^n)\right)^t \left(\py(y^n) \right)^{1-t}\notag\\
&=\sup_t\sum_{y^n} \pyl(y^n)\log \frac{\left( \py\right)_{1-t}(y^n)}{\py(y^n)}\notag\\
&=\sup_t D\left( \pyl\parallel \py \right)-D\left( \pyl \parallel \left(\py\right)_{1-t} \right)\notag\\
&\leq D\left( \pyl\parallel \py \right) =\sum_{y^n} \left( \frac{1}{F_n}\sum_{i=1}^{K_n}\sum_{m=1}^{M_n}\qiml(y^n) \right)\log\left(\frac{\frac{1}{F_n}\sum_{i=1}^{K_n}\sum_{m=1}^{M_n}\qiml(y^n)}{\frac{1}{F_n}\sum_{i=1}^{K_n}\sum_{m=1}^{M_n}\qim(y^n)}\right)\notag\\
&\leq \sum_{y^n} \frac{1}{F_n}\sum_{i=1}^{K_n}\sum_{m=1}^{M_n}\qiml(y^n)\log \frac{\qiml(y^n)}{\qim(y^n)}\label{eq:logsum}
=\frac{1}{F_n}\sum_{i=1}^{K_n}\sum_{m=1}^{M_n}D\left(\qiml\parallel \qim \right) \\
&= \frac{1}{K_n}\sum_{i=1}^{K_n} nD\left( {Q_i}_{\lambda_i}\parallel Q\vert P_i\right),
\end{align}
and where the inequality in~\eqref{eq:logsum} is by Log-Sum inequality.

\subsection{Proof of Lemma~\ref{lemma:no overlapping}}~\label{appendix: no overlapping}
 We will prove the Lemma by contradiction. 

Define \begin{align}t_i \triangleq \text{\textit{Number of users in block }}\, i.\end{align} 
Assume that the arrangement with highest probability (let us call it $\mathcal{A}$) has at least two blocks, say blocks $1,2$, for which $t_1-t_2>1$. This assumption means that the arrangement with the highest probability is {\it not} the non-overlapping arrangement.

The probability of this arrangement, $\PP(\mathcal{A})$, is proportional to  
\begin{align}
\PP(\mathcal{A}) &\propto \binom{K_n}{t_1} \binom{K_n-t_1}{t_2}=\frac{K_n!}{t_1!(K_n-t_1)!}\frac{(K_n-t_1)!}{t_2!(K_n-t_1-t_2)!}\\
&=\frac{K_n!}{t_1! t_2!(K_n-t_1-t_2)!}.
\end{align}

We now consider a new arrangement, $\mathcal{A}_{\text{new}}$,  in which $t_{1,\text{new}}=t_1-1$ and $t_{2,\text{new}}=t_2+1$ and all other blocks remain unchanged. This new arrangement is also feasible since we have not changed the number of users. The probability of this new arrangement is proportional to
\begin{align}
\PP(\mathcal{A}_{\text{new}})& \propto \binom{K_n}{t_1-1} \binom{K_n-t_1-1}{t_2+1}\\
&=\frac{K_n!}{(t_1-1)!(K_n-t_1+1)!}\frac{(K_n-t_1+1)!}{(t_2+1)!(K_n-t_1-t_2)!}\\
&=\frac{K_n!}{(t_1-1)! (t_2+1)!(K_n-t_1-t_2)!}.
\end{align}
Comparing $\PP(\mathcal{A})$ and $\PP(\mathcal{A}_{\text{new}})$ we see that $\PP(\mathcal{A})<\PP(\mathcal{A}_{\text{new}})$ which is a contradiction to  our primary assumption that $\mathcal{A}$ has the highest probability among all arrangements. Hence there do not exist  two blocks which differ more than one in the number of active users within them in the arrangement with the highest probability.

\bibliography{refs}
\bibliographystyle{IEEEtran}

\end{document}